\documentclass[aps,floatfix,amsmath,twocolumn,notitlepage,nofootinbib,longbibliography,superscriptaddress,pra]{revtex4-2}

\usepackage[ruled]{algorithm2e}
\SetKwFor{For}{for (}{) $\lbrace$}{$\rbrace$}
\usepackage{amssymb}
\usepackage{graphicx}
\usepackage{graphics}
\usepackage{amsmath}
\usepackage{amsthm}
\usepackage{color}
\usepackage{dsfont}
\usepackage{wrapfig}

\usepackage{longtable}

\usepackage{mathtools}
\usepackage[colorlinks=true, linkcolor=blue, citecolor=red]{hyperref}
\usepackage{verbatim}
\usepackage{xcolor}

\usepackage{multirow}

\usepackage{wrapfig}
\usepackage{array}

\usepackage{multirow}

\newcolumntype{C}[1]{>{\centering\let\newline\\\arraybackslash\hspace{0pt}}m{#1}}

\newcommand{\p}{\mathrm{p}}
\newcommand{\pr}{\mathrm{PR}}

\usepackage{color}

\def\>{\rangle}
\def\<{\langle}

\def\id{\mathsf{id}}

\newcommand{\ketbra}[2]{|#1\rangle\!\langle#2|}

\renewcommand{\geq}{\geqslant}
\renewcommand{\leq}{\leqslant}

\newcommand{\B}{\mathbf{B}}
\newcommand{\C}{\mathbf{C}}

\newtheorem*{theorem*}{Theorem}

\newtheorem{lemma}{Lemma}
\newtheorem*{lemma*}{Lemma}

\newtheorem{definition}{Definition}[section]
\newtheorem*{definition*}{Definition}

\theoremstyle{definition}

\theoremstyle{remark}

\theoremstyle{definition}

\newtheorem*{axiom*}{Axiom}

\newcommand{\xtens}{\mathbin{\mathop{\otimes}\limits_{\max}}}
\newcommand{\ntens}{\mathbin{\mathop{\otimes}\limits_{\min}}}

\newcommand{\tr}{{\rm Tr}}

\newcommand\independent{\protect\mathpalette{\protect\independenT}{\perp}}
\def\independenT#1#2{\mathrel{\rlap{$#1#2$}\mkern2mu{#1#2}}}

\newcommand{\ve}[1]{{\left\vert\kern-0.25ex\left\vert\kern-0.25ex\left\vert #1 
    \right\vert\kern-0.25ex\right\vert\kern-0.25ex\right\vert}}

\newcommand{\bra}[1]{\langle #1|}
\newcommand{\ket}[1]{|#1\rangle}

\newcommand{\1}{\mathds{1}}

\definecolor{cool_green}{rgb}{0.0, 0.5, 0.0}

\begin{document}

\title{Achieving Maximal Causal Indefiniteness in a Maximally Nonlocal Theory}
\author{Kuntal Sengupta}
\email{kuntal.sengupta.in@gmail.com}
\thanks{A preliminary version of this work appears in the
author's PhD thesis under submission at the University of York.}

\affiliation{Department of Mathematics, University of York, Heslington, York, YO10 5DD, United Kingdom}

\begin{abstract} 

Quantum theory allows for the superposition of causal orders between operations, i.e., for an indefinite causal order; an implication of the principle of quantum superposition. Since a higher theory might also admit this feature, an understanding of superposition and indefinite causal order in a generalised probabilistic framework is needed. We present a possible notion of superposition for such a framework and show that in maximal theories, respecting non-signalling relations, single system state-spaces do not admit superposition; however, composite systems do. Additionally, we show that superposition does not imply entanglement. Next, we provide a concrete example of a maximally Bell-nonlocal theory, which not only admits the presented notion of superposition, but also allows for post-quantum violations of theory-independent inequalities that certify indefinite causal order; even up to an algebraic bound. These findings might point towards potential connections between a theory's ability to admit indefinite causal order,  Bell-nonlocal correlations and the structure of its state spaces. 
\end{abstract}
\maketitle

\section{Introduction}

General relativity(GR) and quantum theory(QT) have been shown to be successful at describing cosmic and atomic physics. A long standing quest is to obtain a higher theory that describes physics at all scales, in a way that its descriptions of cosmic and atomic physics are equivalent to that of GR and QT respectively. In order to investigate for such a theory, it might be useful to work in a framework in which features of both the theories can be expressed meaningfully. From an operational perspective, since QT is probabilistic, it is reasonable to assume that this higher theory is probabilistic as well. 

A following minimal requirement is to be able to describe operations and causal orderings amongst operations in every way permissible in GR and QT. If an operation causally precedes another, we say that there exists a definite causal order relating the two. Here, we are interested in scenarios where there is a lack of definite causal order. It turns out that in both GR and QT such scenarios are possible. On one hand, in GR, if two operations occur in regions that are space-like separated, then there is no definite causal ordering between them. On the other, in QT, the ordering between two operations can be in quantum superposition~\cite{PhysRevA.88.022318}, a feature known as \textit{indefinite causal order(ICO)}. In a probabilistic theory, the lack of causal definiteness of events arising from GR can be modelled by assuming that operations performed in space-like separated regions commute. However, that of superposition of operations remains unknown.  In this paper, we try to address this gap in the framework of generalised probabilistic theories(GPTs)~\cite{PhysRevA.75.032304}.


Fundamental to causal superposition of operations is the notion of superposition. In QT, this notion is attributed to the fact that there exists a representation of pure states in which every pure state can be expressed as a linear combination of other pure states. In an arbitrary probabilistic theory, such a representation need not exist. Therefore, one needs an operational understanding of superposition that can be used to check whether a probabilistic theory admits superposition or not. In this work, we present a candidate definition for superposition which when applied to quantum theory is equivalent to the standard notion of quantum superposition and shows that no classical probability theory can admit superposition. Additionally, we found that although composite systems of maximal theories respecting no-superluminal signalling~\cite{pr} display superposition, their single systems do not.  In addition, we manage to show that the presence of superposition in a theory does not imply the presence of entanglement.

A typical example of indefinite causal order in quantum theory is attributed to a process called the \textit{quantum switch}~\cite{PhysRevA.88.022318}.  That the quantum switch gives rise to indefinite causal order has been tested in a device-independent fashion~\cite{dourdent2023,vanderLugt2023}. In~\cite{vanderLugt2023}, the authors have introduced inequalities, violations of which certify indefinite causal order. Our primary result is an example is a post-quantum probabilistic theory that displays maximally nonlocal correlations, admits superposition and  allows for an indefinite causal order between operations.  We have shown that the inequalities mentioned above can be violated by our toy theory. Such violations are by amounts larger than achievable in QT; even up to an algebraic bound. 

This paper is arranged as follows: Section~\ref{Section::Prelim} contains preliminaries reviewing the quantum switch, the inequalities mentioned in the previous paragraph, a quantum strategy to violate one of those inequalities and a very brief introduction to the framework of GPTs.  In Section~\ref{Sec::Superposition}, we motivate and introduce an operational definition of superposition for probabilistic theories and show how two theories with the identical single system state spaces, differ in their ability to admit superposition. In Section~\ref{Section::HexSquare}, we introduce a foil theory that admits superposition and in Section~\ref{Section::MainResults} show how this theory violates the inequalities reviewed in Section~\ref{Section::Prelim}. Section~\ref{Section::Discussions} contains discussions on the various results.

\section{Preliminaries}
\label{Section::Prelim}
\subsection{Quantum Switch and DRF Inequalities}

\subsubsection{Quantum Switch}

The quantum switch~\cite{PhysRevA.88.022318} is a process in which the causal order between two operations, $\mathcal{O}_{\mathcal{A}_1}$ and  $\mathcal{O}_{\mathcal{A}_2}$, is controlled by a quantum system. Let us assume that this system is a qubit, and $\mathcal{O}_{\mathcal{A}_1}$ precedes  $\mathcal{O}_{\mathcal{A}_2}$ when it is in the state $\ketbra{0}{0}$ and $\mathcal{O}_{\mathcal{A}_2}$ precedes  $\mathcal{O}_{\mathcal{A}_1}$ when in state $\ketbra{1}{1}$. In particular,  entanglement between the state of the control system and the causal orders in which the operations $\mathcal{O}_{\mathcal{A}_1}$ and  $\mathcal{O}_{\mathcal{A}_2}$ are performed is established. In the case where the operations  $\mathcal{O}_{\mathcal{A}_1}$ and  $\mathcal{O}_{\mathcal{A}_2}$ represent unitaries $U_1$ and $U_2$, the action of the quantum switch on them can be defined as 
\begin{equation}
    \left(U_1,U_2\right) \mapsto \ketbra{0}{0}^{\C} \otimes  U_2U_1 + \ketbra{1}{1}^{\C} \otimes  U_1U_2,
\end{equation}
where $\C$ represents the control qubit system~\cite{Dong2023quantumswitchis}. When the control qubit is in a superposition of the states $\ketbra{0}{0}$ and $\ketbra{1}{1}$, the ordering of the operations become superposed, exhibiting an indefinite causal order. That an ICO is displayed by the quantum switch, can be certified in a theory independent way by witnessing violations of certain inequalities~\cite{vanderLugt2023}, that we proceed to review next.

\subsubsection{DRF Inequalities}
\label{Subsection::DRF}

\begin{figure}
    \centering
    \includegraphics[width=0.45\textwidth]{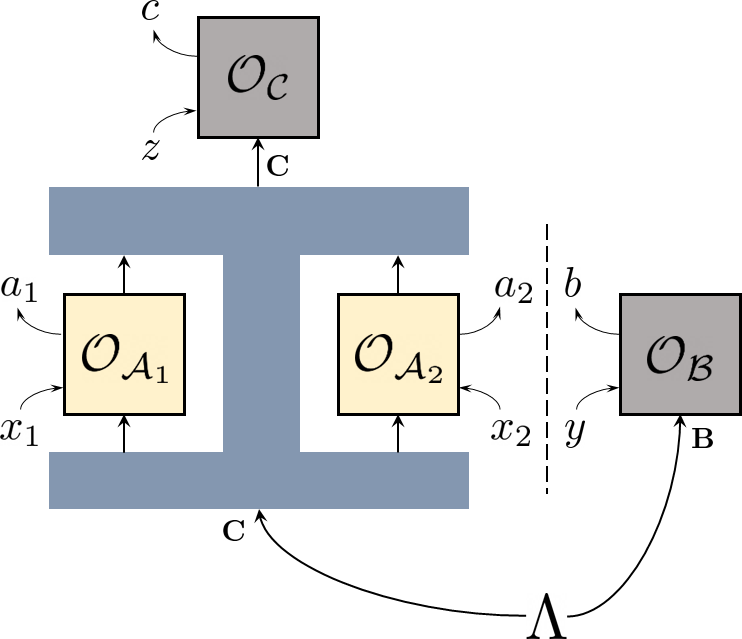}
    \caption{Abstraction of the setup considered in~\cite{vanderLugt2023} to derive the DRF inequalities in Table~\ref{Table:DRF_Faces}. Blue region region represents a process that implements the operations $\mathcal{O}_{\mathcal{A}_1}$ and $\mathcal{O}_{\mathcal{A}_2}$ (yellow) in the orders $\mathcal{O}_{\mathcal{A}_1} \prec \mathcal{O}_{\mathcal{A}_2}$ when the value of subsystem $\C$ of $\Lambda$ is 0 and $\mathcal{O}_{\mathcal{A}_2} \prec \mathcal{O}_{\mathcal{A}_1}$ when the value of  subsystem $\C$ of $\Lambda$ is 1. $\mathcal{O}_{\mathcal{C}}$ is in the causal future of $\{\mathcal{O}_{\mathcal{A}_i}\}_i$ and $\mathcal{O}_{\mathcal{B}}$ is causally disjoint from $\{\mathcal{O}_{\mathcal{A}_i}\}_i$ and $\mathcal{O}_{\mathcal{C}}$. } 
    \label{Fig::DISwitch}
\end{figure}

Given two random variables $A$ and $X$, $A \independent X $  denotes that $A$ is uncorrelated from $X$, i.e., $\p(A|X)=\p(A)$.  For a set, $\Gamma$, of random variables, a causal order, $(\Gamma, \dashrightarrow)$, is a pre-order relation $\dashrightarrow$ on $\Gamma$. For $A,X \in \Gamma$, if $X \dashrightarrow A$, $A$ is said to be in the \textit{causal future} of $X$ and $X$ is said to be a \textit{potential cause} for $A$. $X$ is said to be \textit{free}, if it is uncorrelated with every random variable outside its causal future. To every operation $\mathcal{O}$, we associate a pair of random variables $(M,N)$, with $M \dashrightarrow N$. If $(M',N')$ is associated to the operation $\mathcal{O}'$, the relation $\mathcal{O} \prec \mathcal{O}'$ denotes that the random variables $M'$ and $N'$ are in the causal futures of the random variables $M$ and $N$. The relation $\mathcal{O} \bigtimes \mathcal{O}'$  denotes that neither $\mathcal{O} \prec \mathcal{O}'$ nor  $\mathcal{O}' \prec \mathcal{O}$ hold.

The scenario described in~\cite{vanderLugt2023} considers four operations, $\mathcal{O}_{\mathcal{A}_1},\mathcal{O}_{\mathcal{A}_2}, \mathcal{O}_{\mathcal{B}}$ and $\mathcal{O}_{\mathcal{C}}$, with associated random variable pairs $(X_1,A_1), (X_2,A_2),(Y,B)$ and $(Z,C)$ respectively. The underlying causal relations between these operations can be stated in three parts:

\begin{enumerate}
    \item \textbf{Definite Causal Order (D):}  There is a random variable $\Lambda$ taking values $\lambda \in\{0,1\}$.  When $\lambda =0$,   $\mathcal{O}_{\mathcal{A}_1} \prec \mathcal{O}_{\mathcal{A}_2}$ and when $\lambda = 1$,  $\mathcal{O}_{\mathcal{A}_2} \prec \mathcal{O}_{\mathcal{A}_1}$. 
    \item \textbf{Relativistic Causality (R):} $\mathcal{O}_{\mathcal{A}_i}  \prec \mathcal{O}_{\mathcal{C}}$,  $\mathcal{O}_{\mathcal{A}_{i}} \bigtimes \ \mathcal{O}_{\mathcal{B}}$, for any $i \in \{1,2\}$, and $\mathcal{O}_{\mathcal{C}} \bigtimes \ \mathcal{O}_{\mathcal{B}}$.
    \item \textbf{Freedom of Choice (F):} $X_1,X_2,Y$ and $Z$ are free.
\end{enumerate}

Figure~\ref{Fig::DISwitch} depicts one way to implement these causal relations with the justifications that the operations take place in four closed labs  $\mathcal{A}_1,\mathcal{A}_2,\mathcal{B}$, and $\mathcal{C}$ respectively, the variables $X_1,X_2,Y$ and $Z$ correspond to the input settings available to the operations, and $A_1,A_2,B$ and $C$  correspond to their respective outputs. Under these assumptions, the set of conditional probability distributions, $\p(A_1,A_2,C,B|X_1,X_2,Z,Y)$, can be characterised by a convex polytope, some face-defining inequalities of which are enlisted in Table~\ref{Table:DRF_Faces}. These inequalities hold even if the random variable $\Lambda$ is correlated with the outcome represented by $B$, as depicted in Fig.~\ref{Fig::DISwitch}. Within any theory, admitting these assumptions, whenever the order of operations between $\mathcal{O}_{\mathcal{A}_1}$ and $\mathcal{O}_{\mathcal{A}_2}$ is determined by $\Lambda$, every conditional distribution satisfies inequalities in~\ref{Table:DRF_Faces}. Therefore, these inequalities present theory-independent constraints on conditional probabilities that respect the above assumptions; we will call them \textit{DRF inequalities}.

\begin{table*}[t]
    \centering
    \begin{tabular}{|c|}
    \hline
        \vbox{\begin{equation}
            p(b=0,a_2=x_1|y=0 )+p\left(b=1,a_1=x_2|y=0\right) + p\left(b \oplus c=yz|x_1=x_2=0\right) \leq \frac{7}{4}
      \label{Eq::DRFInequality}  \end{equation}  }  \\     \hline   
      \vbox{\begin{equation}
          p(b=0,a_2=x_1|y=0) + p(b=1,a_1=x_2|y=0) + p(b \oplus c = \boldsymbol{x_2 y}|\boldsymbol{x_1= 0} ) \leq \frac{7}{4}
    \label{Eq::DRFInequality1}  \end{equation}}  \\     \hline 
      \vbox{\begin{equation}
          p(b=0,a_2=x_1|\boldsymbol{x_2}y=\boldsymbol{0}0) + p(b=1,a_1=x_2|\boldsymbol{x_1}y=\boldsymbol{0}0) + p(b \oplus c = x_2 y|x_1= 0 ) \leq \frac{7}{4}
     \label{Eq::DRFInequality2} \end{equation}}  \\     \hline 
      \vbox{\begin{equation}
      \begin{split}
          p(b=0,a_2=x_1|x_2y=00) + p(b=&1,a_1=x_2|x_1y=00) + p(b \oplus c = x_2 y|x_1= 0 )\\ + &\boldsymbol{p(a_2=1,c \oplus1=b=y|x_1x_2=00)} \leq \frac{7}{4} \\
     \end{split} \label{Eq::DRFInequality3} \end{equation}}  \\     \hline 
      \vbox{\begin{equation}
      \begin{split}
          \frac{1}{2}\big[p(a_1=0|x_1x_2=10) + p(&a_2=0|x_1x_2=01) -  p(a_1a_2=00|x_1x_2=11)\big] \\
          &p((x_2a_1+(x_2 \oplus 1)c)\oplus b=x_2y | x_1=x_2) \leq \frac{7}{4} \\
          \end{split}
     \label{Eq::DRFInequality4} \end{equation}}  \\     \hline 
    \end{tabular}
    \caption{Some face-defining hyperplanes of the DRF polytope presented in~\cite{vanderLugt2023}. Therein, authors showed violations of each of these inequalities by quantum switch correlations. Boldface highlights how an inequality is different from the one preceding it and $\oplus$ denotes modulo 2 operation.}
    \label{Table:DRF_Faces}
\end{table*}


\noindent

When the order of the operations $\mathcal{O}_{\mathcal{A}_1}$ and $\mathcal{O}_{\mathcal{A}_2}$ is controlled by one subsystem of a bipartite maximally entangled state, while the other subsystem is distributed to lab $\mathcal{B}$, it is possible to violate the DRF inequalities~\cite{vanderLugt2023}. In particular, the DRF assumptions do not simultaneously hold. If one further assumes that \textbf{R} and \textbf{F} hold, a theory-independent violation of definite causal order (\textbf{D}) is implied. 

Below, we summarize the quantum strategy from~\cite{vanderLugt2023} that leads to a violation of Inequality~\eqref{Eq::DRFInequality}. In Section~\ref{Section::HexSquare}, we will use the quantum operations involved in this strategy to develop our foil theory.


\subsection{Quantum Strategy in the Switch}
\label{Subsection::QStrat}

Let us denote by 
$$
\sigma_{\hat{X}} \coloneqq \left(
\begin{array}{cc}
 0 & 1 \\
 1 & 0 \\
\end{array}
\right) , \sigma_{\hat{Y}} \coloneqq \left(
\begin{array}{cc}
 0 & -i \\
 i & 0 \\
\end{array}
\right) \text{ and } \sigma_{\hat{Z}} \coloneqq  \left(
\begin{array}{cc}
 1 & 0 \\
 0 & -1 \\
\end{array}
\right)
$$
\noindent
 the three Pauli matrices corresponding to some fixed orthogonal directions $\hat{X},\hat{Y}$ and $\hat{Z}$ respectively and let $\ket{\phi_+}\coloneqq (\ket{00}+\ket{11})/\sqrt{2}$. Operations $\mathcal{O}_{\mathcal{A}_i}$ are measure and prepare channels and operations $\mathcal{O}_{\mathcal{B}}$ and $\mathcal{O}_{\mathcal{C}}$ are measurements. One subsystem ($\C$) of a maximally entangled state  $\Phi_+^{\C\B} \coloneqq \ketbra{\phi_+}{\phi_+}^{\C\B}$ is used as the control while the other subsystem ($\B$) is distributed to lab $\mathcal{B}$. A target system, $\mathbf{T}$, is initially prepared in $\ketbra{0}{0}$, on which the operations $\mathcal{O}_{\mathcal{A}_i}$ are performed in the causal order set by the control.  In lab $\mathcal{A}_i$, the incoming system, $\mathbf T$, is measured in the $\{\ketbra{0}{0},\ketbra{1}{1}\}$ basis. The outcome is labelled $a_i$, and the state $\ketbra{x_i}{x_i}$ is prepared and sent off. In lab $\mathcal{C}$, the output control system is measured in the basis generated by the rank 1 projectors of $(\sigma_{\hat{Z}} +  \sigma_{\hat{X}})/\sqrt{2}$ when $z=0$ and of $(\sigma_{\hat{Z}} -  \sigma_{\hat{X}})/\sqrt{2}$ when $z=1$; let us denote these projectors as $\{\ketbra{\psi_{c|z}}{\psi_{c|z}}\}_{c,z}$. The outcome of this measurement is labelled as $c$. In lab $\mathcal{B}$, subsystem $\B$ of $\Phi_+^{\C\B}$ is measured in the basis defined by the rank 1 projectors of $\sigma_{\hat{Z}}$ when $y=0$ and $\sigma_{\hat{X}}$ when $y=1$; let us denote these projectors as $\{\ket{\phi_{b|y}}\bra{\phi_{b|y}}\}_{b,y}$. The outcome of this measurement is labelled as $b$.

 The elements of resultant probability distribution $\p(A_1,A_2,B,C|X_1,X_2,Y,Z)$ can be written as:
\begin{equation}
\label{Eq::ProbFormula}
    \p \left(a_1,a_2,b,c|x_1,x_2,y,z\right) =  \tr\left[ K \left(\Phi_+^{\B\C} \otimes \ket{0}\bra{0}^{\mathbf T}\right) K^{\dagger} \right],
\end{equation} 

\noindent
where $K \coloneqq \bra{\psi_{c|z}}^\C\bra{\phi_{b|y}}^\B(\ket{0}\bra{0}^{\C}\otimes \ket{x_2}\left<a_2|x_1\right>\bra{a_1}^{\mathbf T} +  \ket{1}\bra{1}^{\C} \otimes \ket{x_1}\left<a_1|x_2\right>\bra{a_2}^{\mathbf T})$ and $\tr[\cdot]$ denotes matrix trace.

Note, that when $y=0$, the probability of getting either $b=0$ or $b=1$ is $1/2$. When $b=0$, the post-selected control system is in the state $\ketbra{0}{0}^\C$ which implies $a_2=x_1$. Similarly, when $b=1$, the post-selected control system is in the state $\ketbra{1}{1}^{\C}$ which implies $a_1=x_2$. Therefore, the first two terms of Inequality~\eqref{Eq::DRFInequality} adds up to 1. Next, when $x_1=x_2=0$, the right hand side of Equation~\eqref{Eq::ProbFormula} reduces to
$$
\tr\left[\ketbra{\psi_{c|z}}{\psi_{c|z}}^\C \otimes \ketbra{\phi_{b|y}}{\phi_{b|y}}^\B \left(\Phi_+^{\C\B}  \otimes \ketbra{0}{0}^{\mathbf T}  \right) \right]_{\delta_{a_1=a_2=0}},
$$
\noindent
i.e., the switch process acts as an identity map on the control and target systems. Therefore, labs $\mathcal{B}$ and $\mathcal{C}$ can perform a Bell-test on systems $\B$ and $\C$ to obtain $p\left(b \oplus c=yz|x_1=x_2=0\right)=(1+1/\sqrt{2})/2$. This results in a violation of Inequality~\eqref{Eq::DRFInequality}, since the sum of all the conditional probabilities appearing in it is $1+(1+1/\sqrt{2})/2 > 7/4$~\cite{vanderLugt2023}.

In~\cite{vanderLugt2023}, the maximum quantum violations of Inequalities~\eqref{Eq::DRFInequality1},~\eqref{Eq::DRFInequality2} and~\eqref{Eq::DRFInequality3} were found to be $1.8274 > 7/4$, and for Inequality~\eqref{Eq::DRFInequality4} to be $1+(1+1/\sqrt{2})/2 > 7/4$. We refer the reader to~\cite{vanderLugt2023} for the quantum strategies leading to these values.

\subsection{Generalised Probabilistic Theories (GPTs)}

Generalised Probabilistic Theories are a framework to study theories in which experiments can be decomposed into preparations, operations and measurements of systems~\cite{segal1947,Davies1970,Ludwig1968,Dahn1968,Stolz1971,Mielnik1974,Giles1970,Gudder1973,PhysRevA.75.032304}. In a GPT, a \textit{state space} $\mathcal{S}$, i.e., the set of states available to a system, is mathematically modelled by a compact convex subset of a real vector space $\mathbb{V}$. A positive linear functional in $\mathbb{V}^*$ on $\mathcal{S}$ that maps elements in $\mathcal{S}$ to real numbers in the interval $[0,1]$, is called an \textit{effect}. The \textit{unit effect}, $u$, maps every state to 1 while the \textit{zero effect}, $v_0$, maps every state to 0. A set of effects, $\mathcal{E}$, containing $u$ and $v_0$ is called an \textit{effect space} if for every $e \in \mathcal{E}$, $u-e \in \mathcal{E}$. An effect space $\mathcal{E}$ is a compact convex subset of $\mathbb{V}^*$. The action of an effect $e \in \mathcal{E}$ on a state $s \in \mathcal{S}$ is defined as $e(s) \coloneqq \left<e,s\right>$, where $\left<\cdot,\cdot\right>$ is an inner product on $\mathbb{V}$. A measurement on a state $s$ is characterised by a set of effects $\{e_i\}_i$, such that $\sum_i e_i = u$. Finally, a state (or effect) is said to be \textit{extreme} or \textit{extremal} if it is outside the convex span of the remaining states (or effects).

The state space of a composite system is formed by the composition of the state spaces of its subsystems. Two examples are the \textit{minimal} and \textit{maximal} tensor product compositions. Given two state spaces $\mathcal{S}_1 \subset V_1 $ and $\mathcal{S}_2 \subset V_2$, their minimal tensor product composition is given by
\begin{equation}
    \mathcal{S}_1 \ntens \mathcal{S}_2 \coloneqq \Big\{s_1 \otimes s_2 \ |\ s_1 \in \mathcal{S}_1, s_2 \in \mathcal{S}_2\Big\},
\end{equation}
where $\otimes$ is the tensor product. Their maximal composition is given by 
\begin{equation}
\begin{split}
   & \mathcal{S}_1  \xtens \mathcal{S}_2 \coloneqq \\
   & \Big\{s \in V_1 \otimes V_2\ |\ e_1 \otimes e_2 (s) \in [0,1] \ \forall \  e_1 \in \mathcal{E}_1, e_2 \in \mathcal{E}_2\Big\},
\end{split}     
\end{equation}
where $\mathcal{E}_1$ and $\mathcal{E}_2$ are the largest effect spaces compatible with the state spaces $\mathcal{S}_1$ and $\mathcal{S}_2$ respectively. In these compositions, every multipartite state has the property that when a local measurement is performed on any of its subsystems, the probabilities associated to the various outcomes are independent of operations performed on the remaining subsystems. We say that such correlations are \textit{non-signalling}. The minimal and maximal compositions identify the smallest and largest composite state spaces respecting the non-signalling condition mentioned above. A popular example of a theory constructed in the maximal composition is \textit{box-world (BW)}. Here, joint states display correlations, called \textit{PR correlations}, that maximally violate Bell inequalities~\cite{pr,PhysRevA.75.032304}.

Qubit quantum theory can be phrased as a GPT by identifying its state space with the set of $2 \times 2$ unit-trace positive semi-definite matrices, a compact convex subset of the real vector space, $\mathbb{H}(\mathbb{C}^2)$, of $2 \times 2$ Hermitian matrices. An  effect, $\Pi$, is a $2 \times 2$ positive semi-definite matrix, such that $\mathds{1}-\Pi$ is also positive semi-definite, where $\mathds{1}$ is the $2 \times 2$ identity matrix. The effect space is the largest set of such matrices with $\mathds{1}$ as the unit effect, and forms a compact convex subset of  $\mathbb{H}(\mathbb{C}^2)$. The action of the effects on the states is given by the Hilbert-Schmidt inner product. Since every pure qubit state $\Psi$ is a rank 1 matrix, it can be represented by its eigen-vector $\ket{\psi} \in \mathbb{C}^2$. Qudit quantum theory can be understood in a similar way by identifying its state space with the compact convex subset, of $d \times d$ unit-trace positive semi-definite matrices, of the real vector space of $d \times d$ Hermitian matrices, and so on. Every pure qudit state can also be represented by a vector in $\mathbb{C}^d$. In this theory, pure states are deemed extremal.

\section{Superposition in GPTs}
\label{Sec::Superposition}

Textbook introduction to quantum superposition  is attributed to the fact that certain linear combinations of pure states, when represented as vectors in $\mathbb{C}^d$, is also a pure state. More precisely, for every pure state $\ket{\phi}$, there exists a pair of states $\{\ket{\psi_1},\ket{\psi_2}\}$, a unique linear combination of  which reproduces $\ket{\phi}$, i.e., 
\begin{equation}
    \alpha \ket{\psi_1} + \beta \ket{\psi_2} = \ket{\phi},
\end{equation}
where $\alpha$ and $\beta$ are complex numbers, such that $|\alpha|^2 + |\beta|^2=1$. This notion of superposition cannot be generalised to arbitrary GPTs since it a priory depends on pure states being represented by vectors in $\mathbb{C}^d$. Indeed, even for quantum theory, this notion of superposition is solely dependent on its Hilbert-space formalism. In a different formalism, for instance if all states were represented by probability tables constructed from tomographic data, a clear understanding of quantum superposition is lacking. To build an understanding, that does not depend on the mathematical framework in which the underlying theory is phrased, one might want to take an operational approach and describe it in terms of the input-output statistics obtained upon performing suitable measurements.  

Two attempts to address this were presented in~\cite{PhysRevLett.128.160402} and~\cite{PhysRevA.101.042118}. In~\cite{PhysRevLett.128.160402}, superposition has been treated on equal footing with non-classicality. In particular, any theory with a non-simplicial~\footnote{The state space of any classical probability theory can be described as a simplex.} state space has been said to admit superposition. In~\cite{PhysRevA.101.042118}, superposition has only been explored for theories whose state spaces have infinitely many pure states. In this paper, we take a slightly different approach by first looking at the statistical features of experimental outcomes that are traditionally associated to the presence of superposition in quantum theory and then characterise a minimal condition for a theory to display similar statistical behaviour.

For our first example, let us consider $\ket{\phi}=\ket{0},\ket{\psi_1}=\ket{+}\coloneqq (\ket{0}+\ket{1})/\sqrt{2}$  and $\ket{\psi_2}=\ket{-}\coloneqq (\ket{0}-\ket{1})/\sqrt{2}$, with  which one has 
\begin{equation}
    \frac{1}{\sqrt{2}}(\ket{+}+\ket{-})=\ket{0}.
\end{equation}
\noindent
Looking at these states as elements of the state space, one notices that for the state $\ketbra{0}{0}$, there exists a unique measurement $\{\ketbra{0}{0},\ketbra{1}{1}\}$, such that the outcome of the measurement is deterministic, i.e., $\tr[\ketbra{0}{0}.\ketbra{0}{0}]=1$ and $\tr[\ketbra{1}{1}.\ketbra{0}{0}]=0$. Similarly, for the states $\ketbra{+}{+}$ and $\ketbra{-}{-}$, there exists a measurement, $\{\ketbra{+}{+}, \ketbra{-}{-} \}$, such that outcome of a measurement on these two states is deterministic. However, when the measurement $\{\ketbra{+}{+}, \ketbra{-}{-} \}$ is performed on $\ketbra{0}{0}$, or $\{\ketbra{0}{0},\ketbra{1}{1}\}$ on $\ketbra{+}{+}/\ketbra{-}{-}$, the probability of each outcome is 1/2. Here, the probability of non-deterministic outcomes being 1/2 is not crucial. To see why,  consider $\ket{\phi}=\ket{0},\ket{\psi_1}=\ket{\psi}\coloneqq  \sqrt{2/3} \ket{0}+\sqrt{1/3} \ket{1}$  and $\ket{\psi_2}=\ket{-}$, with which one has
\begin{equation}
    \ket{0}=\frac{\sqrt{3}}{1+\sqrt{2}}\ket{\psi}+\frac{\sqrt{2}}{1+\sqrt{2}}\ket{-}.
\end{equation}
\noindent
For the state $\ketbra{\psi}{\psi}$ there exists a unique measurement $\{\ketbra{\psi}{\psi},\mathds{1}-\ketbra{\psi}{\psi}\}$ which is deterministic with respect to the state $\ketbra{\psi}{\psi}$ but non-deterministic with respect to both  $\ketbra{0}{0}$ and $\ketbra{-}{-}$. In addition,  $\tr[\ketbra{0}{0}.\ketbra{\psi}{\psi}]=2/3 \neq \tr[\ketbra{0}{0}.\ketbra{-}{-}]$.

Translating this into the framework of GPTs, we first note, that to check whether a binary outcome measurement is deterministic with respect to a state or not, it is sufficient to calculate the probability corresponding to one of the outcomes. Secondly, although in qudit quantum theory for every pair of pure states $\ketbra{\phi'_1}{\phi'_1}$ and $\ketbra{\phi'_2}{\phi'_2}$ there exists a unique effect $\ketbra{\phi'_1}{\phi'_1}$  such that $\tr[\ketbra{\phi'_1}{\phi'_1}].\ketbra{\phi'_1}{\phi'_1}=1$ and $\tr[\ketbra{\phi'_1}{\phi'_1}.\ketbra{\phi'_2}{\phi'_2}]$ is either 0 or in the interval $(0,1)$, this may not be true for an arbitrary GPT; there might be multiple extremal effects with this property. 

The discussion above gives us our requirement for a theory to have superposition: the state space of the theory must have three distinct extremal states $s$ and $\{r_1,r_2\}$, and three extremal effects $e_s$ and $\{f_{r_1},f_{r_2}\}$, such that $\left<e_s,s\right>=1$ but $\left<e_s,r_{1/2}\right> \in (0,1)$, and $\left<f_{r_1},r_1\right>=1$ and $\left<f_{r_2},r_2\right>=1$ but $\left<f_{r_{1/2}},s\right> \in (0,1)$. Since we are only considering extremal states, this requirement already distinguishes a superposition from a classical mixture; for any classical mixture of $\alpha r_1 + (1-\alpha) r_2$, $\left<e_s, \alpha r_1 + (1-\alpha) r_2\right> \in (0,1)$, where $\alpha \in [0,1]$. We now formalise our observations into an operational definition for superposition.

\begin{definition}
    Let $\mathcal{S}$ and $\mathcal{E}$ be a state and effect space pair of a GPT and denote by $\mathrm{Extreme}[\mathcal{S}]$ and $\mathrm{Extreme}[\mathcal{E}]$ the set of extreme states in $\mathcal{S}$ and extreme effects in $\mathcal{E}$ respectively. The GPT is said to admit superposition if there exists three distinct states $s,r_1,r_2 \in \mathrm{Extreme}[\mathcal{S}]$ and three effects $e_{s},f_{r_1},f_{r_2} \in \mathrm{Extreme}[\mathcal{E}]$, such that $\left<e_{s},s\right>=1$, $\left<e_{s},r_{j}\right> \in (0,1)$, $\left<f_{r_{j}},r_{j}\right>=1$ and $\left<f_{r_{j}},s\right> \in (0,1)$, for all $j\in\{1,2\}$.   
    \label{Definition::Superposition}
\end{definition}
 \noindent
Note, that the effects $e_{s}$ and $f_{r_j}$ can neither be the zero or the unit effect. Additionally, since the inner product between any extreme state and any extreme effect in a simplicial theory is either 0 or 1, no classical theory can admit superposition. Finally, for qudit quantum theory, say with $d=3$, although one might be able to represent a pure state as a linear combination of three other pure states, a linear combination of two pure states is still a well defined pure state. Thus, Definition~\ref{Definition::Superposition} captures the minimal necessary requirements for a theory to admit superposition.

In a more restrictive notion of superposition, defined with respect to a basis, one might require $\{f_{r_1},f_{r_2}\}$ to form a measurement, i.e., $f_{r_1}+f_{r_2}=u$. In fact, the examples in this paper are constructed to meet this requirement.

It is possible to satisfy the conditions stated in Definition~\ref{Definition::Superposition} by mixed states lying on the boundary of the state space. One might urge that superposition should then be defined for mixed states as well. Indeed, there is no a priory operational reason as to why only extremal states must possess superposition. Therefore, we do not impose a direct restriction on the type of states that might be described as superposition of other states. However, in this work, superposition can be associated to mixed states subject to the theory admitting superposition, in accordance to Definition~\ref{Definition::Superposition}.

Next, let us look at the \textit{gbit} state space consisting of the following set of extremal states:
\begin{equation}
  \Biggl\{  \left(
\begin{array}{c}
 1 \\
 0 \\ \hline
 1 \\
 0 \\
\end{array}
\right), \left(
\begin{array}{c}
 1 \\
 0 \\ \hline
 0 \\
 1 \\
\end{array}
\right), \left(
\begin{array}{c}
 0 \\
 1 \\ \hline
 1 \\
 0 \\
\end{array}
\right),\left(
\begin{array}{c}
 0 \\
 1 \\ \hline
 0 \\
 1 \\
\end{array}
\right)\Biggl\},
\end{equation}
in the notation $\p(A|X) \coloneqq \left(p(0|0) p(1|0)\ |\ p(0|1) p(1|1)\right)^T$, where $X$ and $A$ represent the random variables associated to the choices and outcomes of fiducial measurements~\footnote{A set of measurements is said to be fiducial if they allow for state tomography. It is not necessarily unique. \label{foot}}.  The maximal set of extremal effects for this state space is:

\begin{equation}
   \Biggl\{ \left(
\begin{array}{c}
 1 \\
 0 \\ \hline
 0 \\
 0 \\
\end{array}
\right),\left(
\begin{array}{c}
 0 \\
 1 \\ \hline
 0 \\
 0 \\
\end{array}
\right), \left(
\begin{array}{c}
 0 \\
 0 \\ \hline
 1 \\
 0 \\
\end{array}
\right), \left(
\begin{array}{c}
 0 \\
 0 \\ \hline
 0 \\
 1 \\
\end{array}
\right), \left(
\begin{array}{c}
 0 \\
 0 \\ \hline
 0 \\
 0 \\
\end{array}
\right), \left(
\begin{array}{c}
 1 \\
 1 \\ \hline
 0 \\
 0 \\
\end{array}
\right) \Biggl\}.
\end{equation}
\noindent
Upon constructing a table of inner products between extreme states and effects one can check that all the inner products are either 0 or 1. Therefore, the condition in Definition~\ref{Definition::Superposition} cannot be met, implying that this single system state space does not admit superposition. It is then natural to ask whether there exists a GPT which admits superposition whilst having its single system state spaces being described by gbits. Two examples of theories with gbit single system state spaces are Generalised Local Theory (GLT), formed by taking the min-tensor product of the gbit state spaces, and BW, formed by taking the max-tensor product of the gbit state spaces~\cite{PhysRevA.75.032304}. In a GLT, analogous to the gbit state space, all the inner products between the extremal states and effects are either 0 or 1 and thus GLT is an example of a non-classical theory without superposition. BW, on the other hand,  admits superposition; consider the following collection of states : 
$$
\pr_1 \coloneqq \frac{1}{2}\left(
\begin{array}{cc|cc}
 1 & 0 & 1 & 0 \\
 0 & 1 & 0 & 1 \\  \hline
 1 & 0 & 0 & 1 \\
 0 & 1 & 1 & 0 \\
\end{array}
\right),
\pr_2 \coloneqq \frac{1}{2}\left(
\begin{array}{cc|cc}
 0 & 1 & 0 & 1 \\
 1 & 0 & 1 & 0 \\ \hline
 0 & 1 & 1 & 0 \\
 1 & 0 & 0 & 1 \\
\end{array}
\right),
$$
$$
\pr_1' \coloneqq \frac{1}{2}\left(
\begin{array}{cc|cc}
 1 & 0 & 0 & 1 \\
 0 & 1 & 1 & 0 \\ \hline
 1 & 0 & 1 & 0 \\
 0 & 1 & 0 & 1 \\
\end{array}
\right),
\pr_2' \coloneqq \frac{1}{2}\left(
\begin{array}{cc|cc}
 0 & 1 & 1 & 0 \\
 1 & 0 & 0 & 1 \\ \hline
 0 & 1 & 0 & 1 \\
 1 & 0 & 1 & 0 \\
\end{array}
\right),
$$
\noindent
in the notation $\p(A,B|X,Y) \coloneq $
$$
\left(
\begin{array}{cc|cc}
 p(00|00) & p(01|00) & p(00|01) & p(01|01) \\
 p(10|00) & p(11|00) & p(10|01) & p(11|01) \\ \hline
 p(00|10) & p(01|10) & p(00|11) & p(01|11) \\
 p(10|10) & p(11|10) & p(10|11) & p(11|11) \\
\end{array}
\right)
$$
\noindent
where $X,A$ and $Y,B$ represent the random variables associated to the choices and outcomes of fiducial measurements on the first and second systems respectively, and the collection of effects:
$$
e_1 \coloneqq \left(
\begin{array}{cc|cc}
 0 & 0 & 0 & 0 \\
 0 & 0 & 0 & 0 \\ \hline
 1 & 0 & 0 & 0 \\
 0 & 0 & 1 & 0 \\
\end{array}
\right),
e_2 \coloneqq \left(
\begin{array}{cc|cc}
 0 & 0 & 0 & 0 \\
 0 & 0 & 0 & 0 \\ \hline
 0 & 1 & 0 & 0 \\
 0 & 0 & 0 & 1 \\
\end{array}
\right),
$$
$$
e_1' \coloneqq \left(
\begin{array}{cc|cc}
 0 & 0 & 0 & 1 \\
 0 & 1 & 0 & 0 \\ \hline
 0 & 0 & 0 & 0 \\
 0 & 0 & 0 & 0 \\
\end{array}
\right),
e_2' \coloneqq \left(
\begin{array}{cc|cc}
 0 & 0 & 1 & 0 \\
 1 & 0 & 0 & 0 \\ \hline
 0 & 0 & 0 & 0 \\
 0 & 0 & 0 & 0 \\
\end{array}
\right),
$$
from the bi-partite state and effect spaces. Taking the inner product by element-wise multiplication, we get $\left<e_i,\pr_i\right>=1$, $\left<e_j',\pr_j'\right>=1$, $\left<e_i,\pr_j'\right>=1/2$ and $\left<e_j',\pr_i\right>=1/2$ for any $i,j \in  \{1,2\}$.  Although Definition~\ref{Definition::Superposition} concerns three states and three effects, an additional state and effect in this example helps  observing that  $1/2 \pr_1 + 1/2 \pr_2 = 1/2 \pr_1' + 1/2 \pr_2' $; drawing resemblance to the quantum ensembles of states, $\{\ket{0},\ket{1}\}$ and $\{\ket{+},\ket{-}\}$, discussed above. In the following two sections, we provide an example of a GPT that admits superposition and show how, within it, all the DRF inequalities (see Table~\ref{Table:DRF_Faces}) can be violated by amounts larger than achievable using quantum theory. Finally, note that $e_1+e_2=e_1'+e_2'=u$.

\section{Hex-Square Theory}
\label{Section::HexSquare}

The quantum strategy summarised in Section~\ref{Subsection::QStrat} allows for the maximal violation of Inequality~\eqref{Eq::DRFInequality1}, achievable within quantum theory. The algebraic bound of this inequality is still larger than this value. The part of Inequality~\eqref{Eq::DRFInequality} whose algebraic limit cannot be achieved within quantum theory, is $p\left(b \oplus c=yz|x_1=x_2=0\right)$. The condition $b \oplus c=yz$ can be seen as a nonlocal (CHSH) game, which can be won up to its algebraic maximum by a PR box~\cite{pr,PhysRevA.75.032304} (BW states from Section~\ref{Sec::Superposition}). Therefore, a potential route to outperform quantum theory is to construct a theory where a PR box is a valid bipartite state.  One of its subsystems can then be used to control the causal order between $\mathcal{O}_{\mathcal{A}_1}$ and  $\mathcal{O}_{\mathcal{A}_2}$,  while the other is shared with lab $\mathcal{B}$. One way to represent a PR box is with the state  
\begin{equation}
    \Phi_{\rm PR} \coloneqq \frac{1+\sqrt{2}}{2} \Phi_{+} + \frac{1-\sqrt{2}}{2} \Phi_{-},
    \label{Equation::PR}
\end{equation}
where $\Phi_{-} \coloneqq \ket{\phi_-}\bra{\phi_-}$ with $\ket{\phi_-} \coloneqq (\ket{00}-\ket{11})/\sqrt{2}$. When $\Phi_{\pr}$ is shared between two parties, such that one measures the observables  $\{(\sigma_{\hat X} + \sigma_{\hat Y})/\sqrt{2},(\sigma_{\hat{X}} - \sigma_{\hat Y})/\sqrt{2}\}$ and the other measures $\{\sigma_{\hat X},\sigma_{\hat Y}\}$,  PR correlations are generated~\cite{PhysRevLett.104.140404}.  Note, that $\Phi_{\rm PR}$ is not a quantum state; it has negative eigenvalues.

We present a foil theory in which $\Phi_{\pr}$ is a  valid bipartite state, and whose single system effect spaces contain the effects used in the strategy presented in Section~\ref{Subsection::QStrat}, along with the ones needed to realise PR correlations.  More precisely, the set of extremal effects in the effect space, $\mathcal{E}_{\C}$, of the control subsystem contains the rank 1 projectors of $(\sigma_{\hat X} \pm \sigma_{\hat Z}), (\sigma_{\hat X} \pm \sigma_{\hat Y})$ and $\sigma_{\hat Z}$~\footnote{There is no direct measurement of $\sigma_{\hat{Z}}$. However, in the calculation of probabilities using the formula in Equation~\eqref{Eq::ProbFormula}, the action of the map $K$ on system $\C$ can be seen as a measure and prepare operation onto the eigen-basis of  $\sigma_{\hat{Z}}$.}; that of the other subsystem,  $\mathcal{E}_{\B}$, contains the rank 1 projectors of $\sigma_{\hat X},\sigma_{\hat Y}$ and $\sigma_{\hat Z}$. The remaining extremal effects in both these effect spaces are the unit effect $\mathds{1}$ and the zero effect. The corresponding state spaces $\mathcal{S}_{\C} \subset \mathbb{H}(\mathbb{C}^2)$ and $\mathcal{S}_{\B} \subset \mathbb{H}(\mathbb{C}^2)$ are taken to be the largest set of states compatible with the respective effect spaces. These are explicitly presented below.

Firstly, since the Pauli matrices span $\mathbb{H}(\mathbb{C}^2)$, any unit trace~\footnote{Since $\mathds{1}$ is the unit effect, we require the states to be unit trace.} Hermitian matrix $\varrho$ can be expressed as 
\begin{equation}
    \varrho =\frac{\mathds{1}+ r_x \sigma_{\hat{X}} + r_y \sigma_{\hat{Y}} + r_z \sigma_{\hat{Z}}}{2},
    \label{Eq::PauliState}
\end{equation}
where $r_x,r_y,r_z \in \mathbb{R}$. Second, given an effect space $\mathcal{E} \in \mathbb{V}^*$, the facet-defining inequalities~\footnote{For a polytope $\mathcal{X} \in \mathbb{R}^d$, a set of inequalities $\{\underline{\mathbf{\alpha}}_i^T.\underline{\mathbf{x}} \leq \beta_i\}_i$, with $\underline{\mathbf{\alpha}}_i \in \mathbb{R}^d$ and $\beta_i \in \mathbb{R}$, is called facet-defining if they can be reduced to the set of facet-inequalities.} for the largest compatible state space $\mathcal{S}_{\mathcal{E}} \subset \mathbb{V}$ is given as:  
\begin{equation}
    \mathrm{FD}\left[\mathcal{S}\right] = \bigg\{  \left<e,\mathbf{x} \right> \geq 0 \ |\ e \in \mathrm{Extreme}[\mathcal{E}], \mathbf{x} \in \mathbb{V}  \bigg\},
\end{equation}
where  $\mathrm{Extreme}[\mathcal{E}] $ denotes the extremal effects in $\mathcal{E}$. Since the sets of  extremal effects in both $\mathcal{E}_{\C}$ and $\mathcal{E}_{B}$ are finite, the state spaces for both these systems can be characterised by finite lists of facet-defining inequalities. Finally, calculating the vertices given the list of all facet-defining inequalities of a polytope is known as \textit{vertex enumeration}. For our problem, we enlist the intersection points of the hyperplanes generated by the facet-defining inequalities in variables $(r_x,r_y,r_z)$ and then check which of these intersection points satisfy all such inequalities for the respective systems. In terms of elements in $\mathbb{H}(\mathbb{C}^2)$, we found that the extremal states for systems $\mathbf{C}$ and $\mathbf{B}$ to be:

\begin{equation}
\label{Eq::ExtremeSC}
\begin{split}
    &\mathrm{Extreme}\left[ \mathcal{S}_{\C} \right] = \\    
   & \bigg\{\frac{\mathds{1} \pm \sqrt{2}  \sigma_{\hat{X}}}{2}, \frac{\mathds{1}  \pm \sqrt{2}\sigma_{\hat{Y}} \pm \sigma_{\hat{Z}}}{2}, \frac{\mathds{1} \pm r\sigma_{\hat{X}} \pm \sigma_{\hat{Y}} \pm \sigma_{\hat{Z}}}{2}  \bigg\} \\
\end{split}  
\end{equation}
\noindent
and
\begin{equation}
    \mathrm{Extreme}\left[ \mathcal{S}_{\B} \right] =  \bigg\{\frac{\mathds{1} \pm \sigma_{\hat{X}} \pm \sigma_{\hat{Y}} \pm \sigma_{\hat{Z}}}{2} \bigg\},
\end{equation}
where $r =\sqrt{2}-1$. The square state space associated with system $\mathbf{B}$ can be derived by noticing that the corresponding facet-defining inequalities reduce to $-1 \leq r_x \leq 1$, $-1 \leq r_y \leq 1$ and $-1 \leq r_z \leq 1$. When the rank 1 projectors of $\sigma_{\hat{X}}$, $\sigma_{\hat{Y}}$ and $\sigma_{\hat{Z}}$ are taken to be the fiducial measurements, an isomorphism between $\mathcal{S}_{\B}$ and the gbit state space with three binary outcome fiducial measurements is established. Therefore, system $\mathbf{B}$ does not admit superposition. However, we find that system $\mathbf{C}$ admits superposition, as formalised in Lemma~\ref{Lemma::HexSup}. An outline of the derivation of the extremal set of states of $\mathcal{S}_{\C}$ can be found in Appendix~\ref{App::DerivationExtendedHexSpace}. Figure~\ref{Fig::StateSpaces} represents slices of state spaces $\mathcal{S}_{\mathbf{C}}$ and $\mathcal{S}_{\mathbf{B}}$ obtained by setting $r_y=0$, with respect to the real quantum state space (the Bloch-disc).   

\begin{figure}
    \centering
    \includegraphics[width=0.4\textwidth]{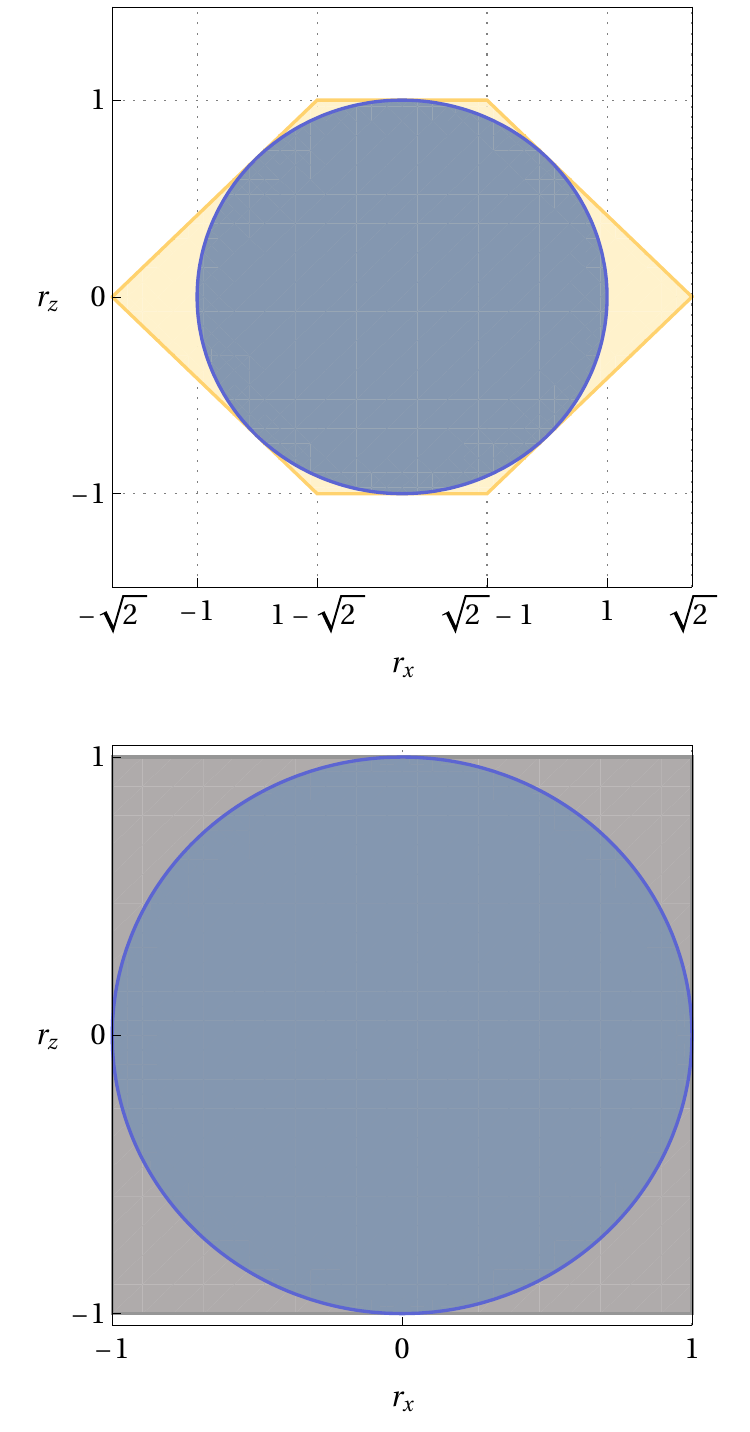}
    \caption{ Slices of state spaces $\mathcal{S}_{\mathbf{C}}$ and $\mathcal{S}_{\mathbf{B}}$ obtained by setting $r_y=0$. \textbf{Top:} Representation of $\mathcal{S}_{\mathbf{C}}$ (yellow) with respect to the quantum set (blue). \textbf{Bottom:} Representation of $\mathcal{S}_{\mathbf{B}}$ (grey) with respect to the quantum set (blue).}
    \label{Fig::StateSpaces}
\end{figure}

\begin{lemma}
    If $\mathcal{S}_{\C}$ and $\mathcal{E}_{\C}$ are a state and effect space pair of a GPT, the GPT admits superposition.
\end{lemma}
\label{Lemma::HexSup}
\begin{proof}

Take the collection of states $\{s_1,s_2,r_1,r_2\}$, where
$$
s_1 \coloneqq \frac{\mathds{1}+\sqrt{2}\sigma_{\hat{X}}}{2}, \quad 
s_2 \coloneqq \frac{\mathds{1}-\sqrt{2}\sigma_{\hat{X}}}{2},
$$
$$
r_1 \coloneqq  \frac{\mathds{1}+ r\sigma_{\hat{X}} +\sigma_{\hat{Y}} +\sigma_{\hat{Z}}}{2}\quad 
r_2 \coloneqq \frac{\mathds{1} - r\sigma_{\hat{X}} -\sigma_{\hat{Y}} -\sigma_{\hat{Z}}}{2},
$$
and effects :
$$
f_1 \coloneqq \frac{\mathds{1}-(\sigma_{\hat{Z}} - \sigma_{\hat{X}})/\sqrt{2}}{2}, 
f_2 \coloneqq \frac{\mathds{1}+(\sigma_{\hat{Z}} - \sigma_{\hat{X}})/\sqrt{2}}{2},
$$
$$
f_1' \coloneqq \frac{\mathds{1}+\sigma_{\hat{Z}}}{2}, \quad 
f_2' \coloneqq \frac{\mathds{1}-\sigma_{\hat{Z}}}{2},
$$
One then gets $\left<f_i,s_i\right>=1$, $\left<f_j',s_j'\right>=1$, $\left<f_i,s_j'\right> \in (0,1)$ and $\left<f_j',s_i\right>=1/2$ for any $i,j \in  \{1,2\}$. 
\end{proof}
\noindent
With the same motive as in BW, we considered four states and effects to draw a closer resemblance to QT: $1/2 s_1 + 1/2 s_2 = 1/2 s_1' + 1/2 s_2' = \mathds{1}/2 $. Lastly, note that $f_1+f_2=f_1'+f_2'=\mathds{1}$.

\subsection{Bipartite Systems}

 Any bipartite composition, $\boxtimes$, of state spaces $\mathcal{S}_{\C}$ and $\mathcal{S}_{\B}$ must have the property that all states within it respect appropriate marginalisations: marginalising to system $\mathbf{C}$ results in a valid state in $\mathcal{S}_{\C}$ and that to $\mathbf{B}$ in $\mathcal{S}_{\B}$. In particular, any unit trace $4 \times 4$ Hermitian matrix, $\varsigma$, is a valid state if $\tr[(\mathrm{P} \otimes \mathrm{Q})\varsigma] \in [0,1]$, for every $\mathrm{P} \in \mathcal{E}_{\mathbf{C}}$ and $\mathrm{Q} \in \mathcal{E}_{\mathbf{B}}$. In addition, for every pair of locally prepared states, there must be a state in the bipartite composition; more precisely, $\mathcal{S}_{\C} \ntens \mathcal{S}_{\B} \subseteq \mathcal{S}_{\C} \boxtimes \mathcal{S}_{\B}$. With these properties, different compositions give rise to different theories. That any such theory admits superposition can also be inferred from its bipartite state space. 

 \begin{lemma}
     Let $\mathcal{S}_{\C} \boxtimes \mathcal{S}_{\B} \subset \mathbb{H(C)} \otimes \mathbb{H(C)}$ be a bipartite composition, $\boxtimes$, of state spaces $\mathcal{S}_{\C}$ and $\mathcal{S}_{\B}$, with the associated effect space $\mathcal{E}_{\C\B} \subset  \mathbb{H(C)} \otimes \mathbb{H(C)}$. There exists three distinct extremal states $s,q_1,q_2 \in \mathcal{S}_{\C} \boxtimes \mathcal{S}_{\B}$ and three distinct extremal effects $e_s,f_{q_1},f_{q_2} \in \mathcal{E}_{\C\B}$, such that $\left<e_{s},s\right>=1$, $\left<e_{s},q_{j}\right> \in (0,1)$, $\left<f_{q_{j}},q_{j}\right>=1$ and $\left<f_{q_{j}},s\right> \in (0,1)$, for all $j\in\{1,2\}$.
     \label{Lemma::CompSup}
 \end{lemma}
\begin{proof}
    First, note that if $s' \in \mathrm{Extreme}[\mathcal{S}_{\B}]$ and $s'' \in \mathrm{Extreme}[\mathcal{S}_{\C}]$, $s'' \otimes s' \in \mathrm{Extreme}[\mathcal{S}_{\C} \boxtimes \mathcal{S}_{\B}]$. Similarly, if  $e' \in \mathrm{Extreme}[\mathcal{E}_{\B}]$ and $e'' \in \mathrm{Extreme}[\mathcal{E}_{\C}]$, $e'' \otimes e' \in \mathrm{Extreme}[\mathcal{E}_{\C\B}]$. Now, pick $s' \in \mathcal{S}_{\B}$ and $e' \in \mathcal{E}_{\B}$ such that $\tr[s'.e']=1$. Further, let $s \coloneqq s_1 \otimes s'$,  $q_1 \coloneqq r_1 \otimes s'$, $q_2 \coloneqq r_2 \otimes s'$, $ e_s \coloneq f_1 \otimes e'$, $f_{q_1} \coloneq f_1' \otimes e'$ and $f_{q_2} \coloneq f_2' \otimes e'$, where $s_1,r_1,r_2,f_1,f_1'$ and $f_2'$ are defined in the proof of Lemma~\ref{Lemma::HexSup}. The conditions in the statement are met with these collection of states and effects.
\end{proof}

Definition~\ref{Definition::Superposition} distinguishes the notion of superposition from entanglement by Lemma~\ref{Lemma::CompSup}. When entanglement is impossible, for instance in $\mathcal{S}_{\C} \ntens \mathcal{S}_{\B}$, superposition is still possible. In particular, superposition does not imply entanglement. However, it is not known to the author whether entanglement implies superposition. 

For our work, we do not specify the bipartite state space explicitly. Upcoming results in Section~\ref{Section::MainResults} hold in any bipartite state space that contains $ \Phi_{\rm PR}$ as a valid state. Any such theory generates all BW correlations. Still, they differ from BW, in that one of the single systems is not described by a gbit state space. We call them \textit{Hex-Square} theories.

\section{Indefinite Causal Order in a Hex-Square Theory}
\label{Section::MainResults}
The main result of this paper is that, within a Hex-Square theory all the DRF inequalities from Table~\ref{Table:DRF_Faces} can be violated by amounts larger than is achievable within quantum theory (see Section~\ref{Subsection::QStrat}). If these amounts are considered indicators of the strength of causal indefiniteness in a theory, a Hex-Square theory can be noted to exhibit post-quantum  causal indefiniteness. 

We found that the maximum achievable value of the expression in Inequality~\eqref{Eq::DRFInequality} is $(14+\sqrt{2})/8$, of that in Inequalities~\eqref{Eq::DRFInequality1},~\eqref{Eq::DRFInequality2} and~\eqref{Eq::DRFInequality3} is $(12+\sqrt{2})/16$, and of that in Inequality~\eqref{Eq::DRFInequality4} is 2. Note, that Inequality~\eqref{Eq::DRFInequality4} can be violated up to its algebraic bound, indicating maximal causal indefiniteness in a Hex-Square theory. In the following, we present the strategies that can be used to obtain these violations.

\subsection{Inequality~\eqref{Eq::DRFInequality} }
\label{Subsection::DRFIneq}

Recall that in the quantum switch the causal order between $\mathcal{O}_{\mathcal{A}_1}$ and $\mathcal{O}_{\mathcal{A}_2}$ is controlled by the states $\ketbra{0}{0}$ and $\ketbra{1}{1}$. Since these states are valid in both the Hex and Square systems, we stick to this convention. Further, recall that these state spaces were constructed by  requiring that the quantum effects used to demonstrate a violation of the DRF Inequality~\eqref{Eq::DRFInequality} are valid effects. Therefore, we will stick to the quantum strategy here as well; with the exception of using a subsystem of $\Phi_{\rm PR}$ instead of $\Phi_+$ for controlling the causal order between the two operations. With this, let us evaluate the values of the probabilities appearing in Inequality~\eqref{Eq::DRFInequality}.

First, when $y=0$ and $b=0$, the post selected sub-normalized state of the control is
\begin{equation}
    \tr_{\B}\left[\left(\id^{\C} \otimes \ket{0}\bra{0}^{\B}\right) \Phi_{\rm PR}^{\C\B}\right] = \frac{1}{2} \ket{0}\bra{0}^{\C}.
\end{equation}
Therefore, the probability of the outcome $b=0$ is 1/2, post-selecting on which the control system is in $\ketbra{0}{0}^{\C}$, implying $a_2=x_1$. Similarly, when $y=0$, with a probability of $1/2$ one gets  $b=1$, with the post-selected state of the control being $\ket{1}\bra{1}^{\C}$, implying $a_1=x_2$.  Hence, the first two terms in Inequality~\eqref{Eq::DRFInequality} add up to 1. For the third term, notice that when $x_1=x_2=0$,

\begin{equation}
\begin{split}
    &\tr\left[ K \left(\Phi_{\rm PR}^{\C\B} \otimes \ket{0}\bra{0}^{\mathbf{T}}\right) K^{\dagger} \right]_{x_1=x_2=0} = \\
    &\frac{1+\sqrt{2}}{2}\tr\left[ K \left(\Phi_{+}^{\C\B} \otimes \ket{0}\bra{0}^{\mathbf{T}}\right) K^{\dagger} \right]_{x_1=x_2=0} + \\
    &\frac{1-\sqrt{2}}{2}\tr\left[ K \left(\Phi_{-}^{\C\B} \otimes \ket{0}\bra{0}^{\mathbf{T}}\right) K^{\dagger} \right]_{x_1=x_2=0} = \\    
    &\tr\left[ \Pi^{\C\B} \otimes \id^{\mathbf{T}} \left(\Phi_{\rm PR}^{\C\B} \otimes \ket{0}\bra{0}^{\mathbf{T}}\right) \right]\delta_{a_1=a_2=0};
    \end{split}  
    \label{Eq::EffIdentity}
\end{equation}
where $\Pi^{\C\B} \coloneqq \ket{\psi_{c|z}}\bra{\psi_{c|z}}^{\C} \otimes \ket{\phi_{b|y}}\bra{\phi_{b|y}}^{\B}$, meaning the operations inside the switch act as an identity on the control and target input systems when $x_1=x_2=0$. Therefore, a Bell-test can be performed in labs $\mathcal{C}$ and $\mathcal{B}$ on $\Phi_{\rm PR}$ to generate the conditional probability distribution
\begin{equation}
    \p(C,B|X,Z)=\left(
\begin{array}{cc|cc}
 \varepsilon_+/8   & \varepsilon_-/8   & 1/2 & 0 \\
 \varepsilon_-/8  & \varepsilon_+/8   & 0 & 1/2 \\ \hline
 \varepsilon_+/8  & \varepsilon_-/8   & 0 & 1/2 \\
 \varepsilon_-/8   & \varepsilon_+/8   & 1/2 & 0 \\
\end{array}
\right)
\end{equation}
where $\varepsilon_\pm = 2 \pm \sqrt{2}$. This distribution has a CHSH score of $(6+\sqrt{2})/8$. The three terms of Inequality~\eqref{Eq::DRFInequality} therefore add up to $(14+\sqrt{2})/8$, which is larger than $1+(1+1/\sqrt{2})/2$, i.e., the maximal violation achievable in quantum theory. Under the assumptions taken in~\cite{vanderLugt2023}, this violation certifies indefinite causal order in a Hex-Square theory. Although a Hex-Square theory can generate all Box-world correlations, with the strategy above we do not see the algebraic maximal violation of Inequality~\eqref{Eq::DRFInequality}. This is because one needs to measure the observables   $(\sigma_{\hat X} \pm \sigma_{\hat Y})/\sqrt{2}$ on system $\C$ and $\{\sigma_{\hat X},\sigma_{\hat Y}\}$ on system $\B$. If either $\sigma_{\hat X}$  or $\sigma_{\hat Y}$ is measured on system $\B$, there are no outcomes, post-selecting on which a definite causal order is achieved between $\mathcal{O}_{\mathcal{A}_1}$ and  $\mathcal{O}_{\mathcal{A}_2}$. Therefore, the sum of the first two terms in Inequality~\eqref{Eq::DRFInequality} cannot be maximised. Alternatively, if the CHSH term is maximised first, an overall lower value will be observed.

\subsection{Inequalities~\eqref{Eq::DRFInequality1},~\eqref{Eq::DRFInequality2} and~\eqref{Eq::DRFInequality3}}
\label{Subsection::DRFIneq1to3}

Inequality~\eqref{Eq::DRFInequality1} is different from~\eqref{Eq::DRFInequality} in that $z$ is replaced by $x_2$. With the strategy from the previous section, the sum of the first two terms is 1. Moreover, when $x_2=y=0$, the probability that $b=0$ and $a_2=x_1$ is 1/2; similarly, when  $x_1=y=0$, the probability of $b=1$ and $a_1=x_2$ is 1/2. Therefore, with the same strategy, the sums of the first two terms in both~\eqref{Eq::DRFInequality2} and~\eqref{Eq::DRFInequality3} are 1. Since the third terms in Inequalities~\eqref{Eq::DRFInequality1}, ~\eqref{Eq::DRFInequality2} and~\eqref{Eq::DRFInequality3} are equal, we need to find a strategy that maximises this probability while keeping the sum of the first two terms 1.

When $x_2=1$, the marginal distribution over $a_1$ and $b$, conditioned on $x_1$ and $y$ is
\begin{equation}    
\begin{split}
    p(a_1,b|x_1,x_2=&1,y) =\\ &\tr\left[ \ketbra{a_1}{a_1}^\C \otimes \ketbra{\phi_{b|y}}{\phi_{b|y}}^\B \Phi_{\pr}^{\C\B}  \right].\\
\end{split}            \label{Equation::EffectiveMes}
\end{equation}
Notice, that setting $x_2=1$ allows for an effective measurement, $\{\ketbra{0}{0},\ketbra{1}{1}\}$, on the control system. On the other hand, from  Equation~\eqref{Eq::EffIdentity}, when $x_1=x_2=0$, the operations inside the switch act as an identity map on the control and target, allowing lab $\mathcal{C}$ to perform an explicit measurement on the control system. Concisely,  when $x_2=0$, lab $\mathcal{C}$ performs a binary outcome measurement on the control, and when $x_2=1$ an effective projective measurement on the rank 1 projectors of $\sigma_{\hat{Z}}$ is performed. The outcomes of both these measurements can be announced by lab $\mathcal{C}$ by setting $c \coloneqq x_2a_1+(x_2 \oplus 1)c'$, where $c'$ is the outcome of the measurement explicitly performed in lab $\mathcal{C}$. This is possible since $\mathcal{O}_{\mathcal{A}_i} \prec \mathcal{O}_{\mathcal{C}}$.  Finally, in lab $\mathcal{B}$, when $y=1$, a measurement should be chosen, such that with the measurement explicitly performed in lab $\mathcal{C}$, the probability of winning the CHSH game $b \oplus c = x_2 y$ is maximised. Upon optimising over measurements constructed from extremal effects, a maximum value of $(12+\sqrt{2})/16$ can be obtained when the measurement in lab $\mathcal{C}$ is formed from the rank 1 projectors of $(\sigma_{\hat{Z}}-\sigma_{\hat{X}})/\sqrt{2}$ and  that in lab $\mathcal{B}$ from the rank 1 projectors of $\sigma_{\hat{X}}$. This value is larger than that achievable using quantum correlations in the strategy mentioned in~\cite{vanderLugt2023}; hence a larger than quantum violation of Inequalities~\eqref{Eq::DRFInequality1} and~\eqref{Eq::DRFInequality2} is achievable. Note, that the value of the final term in Inequality~\eqref{Eq::DRFInequality3} is zero, since $p(a_2=0|x_1x_2=00)=0$. Therefore, its violation is equal to that of Inequality~\eqref{Eq::DRFInequality2}.

Non-extremal measurements do not lead to higher violations since they generate probability distributions that can be expressed as a convex combination of probability distributions generated by extremal measurements.

\subsection{Inequality~\eqref{Eq::DRFInequality4}}
\label{Subsection::DRFIneq4}

So far, the control system state space was $\mathcal{S}_{\C}$ and the the shared system with lab $\mathcal{B}$ was $\mathcal{S}_{\B}$ and the causal order was controlled by the states $\ketbra{0}{0}$ and $\ketbra{1}{1}$. To demonstrate a violation of Inequality~\eqref{Eq::DRFInequality4}, we will set the control system state space to be $\mathcal{S}_{\C'} \coloneqq \mathcal{S}_{\B}$ and the system held by lab $\mathcal{B}$ to be $\mathcal{S}_{\B'}=\mathcal{S}_{\C}$. Additionally, we use rank 1 projectors of $\sigma_{\hat{X}}$ i.e., $\ketbra{\pm}{\pm}$, to control the order of the operations $\{\mathcal{O}_{\mathcal{A}_i}\}_i$. The target system is initialised to the state $\ketbra{+}{+}$. 

In lab $\mathcal{A}_i$, the measurement $\{\ketbra{+}{+},\ketbra{-}{-}\}$ is performed on the incoming target system. The outcome is labelled $a_i$, where $a_i$ is 0 or 1 depending on whether the first or the second outcome is obtained. Next, the state  $\ketbra{+}{+}$ or $\ketbra{-}{-}$ is prepared depending on whether $x_i = 0$ or $1$ respectively, and sent off. In lab $\mathcal{C}$, the output control system is measured in the basis generated by the rank 1 projectors of  $\sigma_{\hat{Y}}$. In lab $\mathcal{B}$, the shared subsystem of $\Phi_{\pr}^{\C'\B'}$ is measured in the rank 1 projectors of $(\sigma_{\hat{X}} +  \sigma_{\hat{Y}})/\sqrt{2}$ when $y=0$ and of $(\sigma_{\hat{X}} - \sigma_{\hat{Y}})/\sqrt{2}$ when $y=1$.

With the setup described, one finds $p(a_1=0|x_1x_2=10)=p(a_2=0|x_1x_2=01)=1$ and $p(a_1a_2=00|x_1x_2=11)=0$. Therefore the sum of the first three terms is 1, its algebraic maximum. For the final term, when $x_1=x_2=1$, the marginal distribution over $a_1$ and $b$ is given by Equation~\eqref{Equation::EffectiveMes}, resulting in the effective measurement,  $\{\ketbra{+}{+},\ketbra{-}{-}\}$, of the control system; when $x_1=x_2=0$, the operations within the switch act as an identity on the control and target systems, allowing lab $\mathcal{C}$ to perform an explicit measurement on the control system. This can be verified by identifying that $\ket{\phi_{\pm}}=(\ket{++} \pm \ket{--})/\sqrt{2}$. Therefore, using the strategy described in the previous section, it is possible to perfectly win the CHSH game $(x_2a_1+(x_2 \oplus 1)c)\oplus b=x_2y$, since the combination of measurements used are precisely the ones that generate PR box correlations, when performed on $\Phi_{\pr}^{\C'\B'}$. Hence, the sum of the terms in Inequality~\eqref{Eq::DRFInequality4} evaluates to 2, its algebraic bound.

\section{Discussion}
\label{Section::Discussions}

We have shown that there exists maximally nonlocal theories that display maximal causal indefiniteness. In particular, it is possible to theory-independently certify indefinite causal order in a Hex-Square theory. In addition, if one were to take the violation of the DRF inequalities as a measure of indefinite causal order, in analogy to Bell inequalities and nonlocality, a larger than quantum violation of the DRF inequalities in a Hex-Square theory posits the presence of stronger than quantum causal indefiniteness. Here, reaching the algebraic limit of Inequality~\eqref{Eq::DRFInequality4} suggests that that quantum theory is neither the most nonlocal nor the most causally indefinite. This opens a new avenue to investigate the constraints on quantum correlations in their ability to display non-classical behaviours. To single out quantum theory from post-quantum GPTs, one might then want to device an information processing task whose optimal performance is reached using quantum correlations generated in an indefinite causal order. This might point towards a way towards possible axiomatisation of quantum theory.

We have shown that Inequalities~\eqref{Eq::DRFInequality},~\eqref{Eq::DRFInequality1},~\eqref{Eq::DRFInequality2} and~\eqref{Eq::DRFInequality3} are not violated up to their algebraic bounds by a Hex-Square theory. The terms of these inequalities that could not be maximised correspond to various CHSH games. If multiple copies of the state space were available, multi-copy nonlocality distillation is possible~\cite{PhysRevLett.102.120401}. Whether these CHSH games can then be perfectly won is not known to the author. In addition, the fact that a maximally nonlocal theory is needed to maximally violate Inequality~\eqref{Eq::DRFInequality4}, might point towards a connection between nonlocality and indefinite causality.

A notable outcome of our work is the presence of superposition in the absence of entanglement. A next natural question is whether entanglement implies superposition. For this, an operational definition of entanglement is needed.

\section{Acknowledgements}
I am grateful to Roger Colbeck for stimulating discussions and to Cyril Branciard, Tein van der Lugt and Vincenzo Fiorentino for useful feedback. This work was supported by the Departmental Studentship from the Department of Mathematics, University of York.

\onecolumngrid
\appendix

\section{Derivation of Hex State Space}
\label{App::DerivationExtendedHexSpace}

\begin{table}[h]
    \centering
    \begin{tabular}{|c|c|c|c|c|}
    \hline
         $ \phantom{\bigg\{}\frac{\mathds{1} + (\sigma_{\hat{X}} + \sigma_{\hat{Z}})/\sqrt{2} }{2} \phantom{\bigg\{}$  & $\phantom{\bigg\{} \frac{\mathds{1} - (\sigma_{\hat{X}} + \sigma_{\hat{Z}})/\sqrt{2} }{2} \phantom{\bigg\{} $  & $\phantom{\bigg\{} \frac{\mathds{1} + (\sigma_{\hat{X}} - \sigma_{\hat{Z}})/\sqrt{2} }{2} \phantom{\bigg\{}$  & $\phantom{\bigg\{} \frac{\mathds{1} - (\sigma_{\hat{X}} - \sigma_{\hat{Z}})/\sqrt{2} }{2} \phantom{\bigg\{}$  & $\frac{\mathds{1} + \sigma_{\hat{Z}} }{2}$  \\ \hline 
         $ \phantom{\bigg\{} \frac{1}{4} \left(\sqrt{2} r_x +\sqrt{2} r_z +2\right)  \phantom{\bigg\{}$ & $  \phantom{\bigg\{}\frac{1}{4} \left(-\sqrt{2} r_x -\sqrt{2} r_z +2\right)  \phantom{\bigg\{}$  & $ \phantom{\bigg\{} \frac{1}{4} \left(-\sqrt{2} r_x +\sqrt{2} r_z +2\right) \phantom{\bigg\{} $  & $ \phantom{\bigg\{} \frac{1}{4} \left(\sqrt{2} r_x-\sqrt{2} r_z+2\right)  \phantom{\bigg\{} $  &$ \phantom{\bigg\{} \frac{r_z +1}{2} \phantom{\bigg\{} $  \\ \hline \hline
           $  \phantom{\bigg\{} \frac{\mathds{1} + (\sigma_{\hat{X}} + \sigma_{\hat{Y}})/\sqrt{2} }{2} \phantom{\bigg\{} $  & $\phantom{\bigg\{} \frac{\mathds{1} - (\sigma_{\hat{X}} + \sigma_{\hat{Y}})/\sqrt{2} }{2} \phantom{\bigg\{}$  & $\phantom{\bigg\{} \frac{\mathds{1} + (\sigma_{\hat{X}} - \sigma_{\hat{Y}})/\sqrt{2} }{2} \phantom{\bigg\{}$  & $\phantom{\bigg\{} \frac{\mathds{1} - (\sigma_{\hat{X}} - \sigma_{\hat{Y}})/\sqrt{2} }{2} \phantom{\bigg\{}$  & $ \phantom{\bigg\{} \frac{\mathds{1} - \sigma_{\hat{Z}} }{2} \phantom{\bigg\{}$   \\ \hline 
           $ \phantom{\bigg\{} \frac{1}{4} \left(\sqrt{2} r_x +\sqrt{2} r_y +2\right)  \phantom{\bigg\{} $ & $ \phantom{\bigg\{} \frac{1}{4} \left(-\sqrt{2} r_x -\sqrt{2} r_y +2\right)  \phantom{\bigg\{}$  & $ \phantom{\bigg\{} \frac{1}{4} \left(\sqrt{2} r_x-\sqrt{2} r_y +2\right)  \phantom{\bigg\{}$  &  $ \phantom{\bigg\{} \frac{1}{4} \left(-\sqrt{2} r_x +\sqrt{2} r_y +2\right) \phantom{\bigg\{}$  &  $ \phantom{\bigg\{}\frac{1-r_z}{2} \phantom{\bigg\{}$  \\ \hline
    \end{tabular}
    \caption{Inner products between the extreme effects of $\mathcal{E}_{\mathrm{C}}$ and the hermitian matrix $\varrho$~\eqref{Eq::PauliState}. Top rows of each block list the extremal effects while the bottom rows list the respective inner products.}
    \label{Tab:InnerExtHex}
\end{table}

One needs that every inner product listed in the entries of the Table~\ref{Tab:InnerExtHex} is between 0 and 1. This puts constraints on the ranges of the variables $r_x,r_y$ and $r_z$. Using the Reduce function of MATHEMATICA, we obtain the reductions provided in the Table~\ref{Tab:ConstraintsExtHex} of these constraints. From the first two constraints, one gets the states 
$$
\left(
\begin{array}{cc}
 \frac{1}{2} & -\frac{1}{\sqrt{2}} \\
 -\frac{1}{\sqrt{2}} & \frac{1}{2} \\
\end{array}
\right), 
\left(
\begin{array}{cc}
 \frac{1}{2} & \frac{1}{\sqrt{2}} \\
 \frac{1}{\sqrt{2}} & \frac{1}{2} \\
\end{array}
\right).
$$

\begin{table}[h]
    \centering
    \begin{tabular}{|c|} \hline
          $\phantom{\bigg\{} r_x =-\sqrt{2}\land r_y =0\land r_z =0 \phantom{\bigg\{}$ \\ \hline
           $\phantom{\bigg\{} r_x =\sqrt{2}\land r_y =0\land r_z =0 \phantom{\bigg\{}$  \\ \hline
        $\phantom{\bigg\{} -\sqrt{2}<r_x\leq \frac{\sqrt{2}-2}{\sqrt{2}}\land \frac{-\sqrt{2} r_x-2}{\sqrt{2}}\leq r_y\leq \frac{\sqrt{2} r_x+2}{\sqrt{2}}\land \frac{-\sqrt{2}
   r_x-2}{\sqrt{2}}\leq r_z\leq \frac{\sqrt{2} r_x +2}{\sqrt{2}} \phantom{\bigg\{}$ \\ \hline
       $\phantom{\bigg\{} \frac{\sqrt{2}-2}{\sqrt{2}}<r_x\leq 0\land \frac{-\sqrt{2} r_x-2}{\sqrt{2}}\leq r_y\leq \frac{\sqrt{2} r_x+2}{\sqrt{2}}\land -1\leq r_z\leq 1 \phantom{\bigg\{}$  \\ \hline
       $\phantom{\bigg\{} 0<r_x\leq \frac{2-\sqrt{2}}{\sqrt{2}}\land \frac{\sqrt{2} r_x-2}{\sqrt{2}}\leq r_y\leq \frac{2-\sqrt{2} r_x}{\sqrt{2}}\land -1\leq r_z\leq 1 \phantom{\bigg\{}$  \\ \hline
        $\phantom{\bigg\{} \frac{2-\sqrt{2}}{\sqrt{2}}<r_x<\sqrt{2}\land \frac{\sqrt{2} r_x-2}{\sqrt{2}}\leq r_y\leq \frac{2-\sqrt{2} r_x}{\sqrt{2}}\land \frac{\sqrt{2}
   r_x-2}{\sqrt{2}}\leq r_z\leq \frac{2-\sqrt{2} r_x}{\sqrt{2}} \phantom{\bigg\{}$ \\ \hline      
    \end{tabular}
    \caption{Constraints on the variables $r_x,r_y$ and $r_z$ upon requiring every inner product listed in Table~\ref{Tab:InnerExtHex} to be in between 0 and 1.}
    \label{Tab:ConstraintsExtHex}
\end{table}
\noindent

From the third constraint onwards, one can set the inequalities to equalities, where possible and solve for the remaining states. We provide these states below, with $r=\sqrt{2}-1$:

$$
\left(
\begin{array}{cc}
 1 & \frac{1}{2} \left(r-i\right) \\
 \frac{1}{2} \left(r+i\right) & 0 \\
\end{array}
\right),
\left(
\begin{array}{cc}
 0 & \frac{1}{2} \left(r-i\right) \\
 \frac{1}{2} \left(r+i\right) & 1 \\
\end{array}
\right),\left(
\begin{array}{cc}
 1 & \frac{1}{2} \left(r+i\right) \\
 \frac{1}{2} \left(r-i\right) & 0 \\
\end{array}
\right),\left(
\begin{array}{cc}
 0 & \frac{1}{2} \left(r+i\right) \\
 \frac{1}{2} \left(r-i\right) & 1 \\
\end{array}
\right),
$$

$$
\left(
\begin{array}{cc}
 1 & \frac{1}{2} \left(-r-i\right) \\
 \frac{1}{2} \left(-r+i\right) & 0 \\
\end{array}
\right),
\left(
\begin{array}{cc}
 0 & \frac{1}{2} \left(-r-i\right) \\
 \frac{1}{2} \left(-r+i\right) & 1 \\
\end{array}
\right),\left(
\begin{array}{cc}
 1 & \frac{1}{2} \left(-r+i\right) \\
 \frac{1}{2} \left(-r-i\right) & 0 \\
\end{array}
\right),\left(
\begin{array}{cc}
 0 & \frac{1}{2} \left(-r+i\right) \\
 \frac{1}{2} \left(-r-i\right) & 1 \\
\end{array}
\right),
$$
$$
\left(
\begin{array}{cc}
 0 & \frac{1}{\sqrt{2}} i \\
 -\frac{1}{\sqrt{2}} i & 1 \\
\end{array}
\right),
\left(
\begin{array}{cc}
 1 & \frac{1}{\sqrt{2}} i \\
 -\frac{1}{\sqrt{2}} i & 0 \\
\end{array}
\right),
\left(
\begin{array}{cc}
 0 & -\frac{1}{\sqrt{2}} i \\
 \frac{1}{\sqrt{2}} i & 1 \\
\end{array}
\right),
\left(
\begin{array}{cc}
 1 & -\frac{1}{\sqrt{2}} i \\
 \frac{1}{\sqrt{2}} i & 0 \\
\end{array}
\right)
$$
The group of first 4 states can be written as $\frac{\mathds{1} \pm \sqrt{2}  \sigma_{\hat{X}}}{2}$. The next 8 states can be written as $\frac{\mathds{1}  \pm \sqrt{2}\sigma_{\hat{Y}} \pm \sigma_{\hat{Z}}}{2}$. The last 4 states can be written as $ \frac{\mathds{1} \pm r\sigma_{\hat{X}} \pm \sigma_{\hat{Y}} \pm \sigma_{\hat{Z}}}{2}$.


\begin{thebibliography}{18}%
\makeatletter
\providecommand \@ifxundefined [1]{%
 \@ifx{#1\undefined}
}%
\providecommand \@ifnum [1]{%
 \ifnum #1\expandafter \@firstoftwo
 \else \expandafter \@secondoftwo
 \fi
}%
\providecommand \@ifx [1]{%
 \ifx #1\expandafter \@firstoftwo
 \else \expandafter \@secondoftwo
 \fi
}%
\providecommand \natexlab [1]{#1}%
\providecommand \enquote  [1]{``#1''}%
\providecommand \bibnamefont  [1]{#1}%
\providecommand \bibfnamefont [1]{#1}%
\providecommand \citenamefont [1]{#1}%
\providecommand \href@noop [0]{\@secondoftwo}%
\providecommand \href [0]{\begingroup \@sanitize@url \@href}%
\providecommand \@href[1]{\@@startlink{#1}\@@href}%
\providecommand \@@href[1]{\endgroup#1\@@endlink}%
\providecommand \@sanitize@url [0]{\catcode `\\12\catcode `\$12\catcode `\&12\catcode `\#12\catcode `\^12\catcode `\_12\catcode `\%12\relax}%
\providecommand \@@startlink[1]{}%
\providecommand \@@endlink[0]{}%
\providecommand \url  [0]{\begingroup\@sanitize@url \@url }%
\providecommand \@url [1]{\endgroup\@href {#1}{\urlprefix }}%
\providecommand \urlprefix  [0]{URL }%
\providecommand \Eprint [0]{\href }%
\providecommand \doibase [0]{https://doi.org/}%
\providecommand \selectlanguage [0]{\@gobble}%
\providecommand \bibinfo  [0]{\@secondoftwo}%
\providecommand \bibfield  [0]{\@secondoftwo}%
\providecommand \translation [1]{[#1]}%
\providecommand \BibitemOpen [0]{}%
\providecommand \bibitemStop [0]{}%
\providecommand \bibitemNoStop [0]{.\EOS\space}%
\providecommand \EOS [0]{\spacefactor3000\relax}%
\providecommand \BibitemShut  [1]{\csname bibitem#1\endcsname}%
\let\auto@bib@innerbib\@empty
\bibitem [{\citenamefont {Chiribella}\ \emph {et~al.}(2013)\citenamefont {Chiribella}, \citenamefont {D'Ariano}, \citenamefont {Perinotti},\ and\ \citenamefont {Valiron}}]{PhysRevA.88.022318}%
  \BibitemOpen
  \bibfield  {author} {\bibinfo {author} {\bibfnamefont {G.}~\bibnamefont {Chiribella}}, \bibinfo {author} {\bibfnamefont {G.~M.}\ \bibnamefont {D'Ariano}}, \bibinfo {author} {\bibfnamefont {P.}~\bibnamefont {Perinotti}},\ and\ \bibinfo {author} {\bibfnamefont {B.}~\bibnamefont {Valiron}},\ }\bibfield  {title} {\bibinfo {title} {Quantum computations without definite causal structure},\ }\href {https://doi.org/10.1103/PhysRevA.88.022318} {\bibfield  {journal} {\bibinfo  {journal} {Phys. Rev. A}\ }\textbf {\bibinfo {volume} {88}},\ \bibinfo {pages} {022318} (\bibinfo {year} {2013})}\BibitemShut {NoStop}%
\bibitem [{\citenamefont {Barrett}(2007)}]{PhysRevA.75.032304}%
  \BibitemOpen
  \bibfield  {author} {\bibinfo {author} {\bibfnamefont {J.}~\bibnamefont {Barrett}},\ }\bibfield  {title} {\bibinfo {title} {Information processing in generalized probabilistic theories},\ }\href {https://doi.org/10.1103/PhysRevA.75.032304} {\bibfield  {journal} {\bibinfo  {journal} {Phys. Rev. A}\ }\textbf {\bibinfo {volume} {75}},\ \bibinfo {pages} {032304} (\bibinfo {year} {2007})}\BibitemShut {NoStop}%
\bibitem [{\citenamefont {Popescu}\ and\ \citenamefont {Rohrlich}(1994)}]{pr}%
  \BibitemOpen
  \bibfield  {author} {\bibinfo {author} {\bibfnamefont {S.}~\bibnamefont {Popescu}}\ and\ \bibinfo {author} {\bibfnamefont {D.}~\bibnamefont {Rohrlich}},\ }\bibfield  {title} {\bibinfo {title} {Nonlocality as an axiom},\ }\href {https://doi.org/10.1007/BF02058098} {\bibfield  {journal} {\bibinfo  {journal} {Foundations of Physics}\ }\textbf {\bibinfo {volume} {24}},\ \bibinfo {pages} {379} (\bibinfo {year} {1994})}\BibitemShut {NoStop}%
\bibitem [{\citenamefont {Dourdent}\ \emph {et~al.}(2023)\citenamefont {Dourdent}, \citenamefont {Abbott}, \citenamefont {Šupić},\ and\ \citenamefont {Branciard}}]{dourdent2023}%
  \BibitemOpen
  \bibfield  {author} {\bibinfo {author} {\bibfnamefont {H.}~\bibnamefont {Dourdent}}, \bibinfo {author} {\bibfnamefont {A.~A.}\ \bibnamefont {Abbott}}, \bibinfo {author} {\bibfnamefont {I.}~\bibnamefont {Šupić}},\ and\ \bibinfo {author} {\bibfnamefont {C.}~\bibnamefont {Branciard}},\ }\href {https://arxiv.org/abs/2308.12760} {\bibinfo {title} {Network-device-independent certification of causal nonseparability}} (\bibinfo {year} {2023}),\ \Eprint {https://arxiv.org/abs/2308.12760} {arXiv:2308.12760 [quant-ph]} \BibitemShut {NoStop}%
\bibitem [{\citenamefont {van~der Lugt}\ \emph {et~al.}(2023)\citenamefont {van~der Lugt}, \citenamefont {Barrett},\ and\ \citenamefont {Chiribella}}]{vanderLugt2023}%
  \BibitemOpen
  \bibfield  {author} {\bibinfo {author} {\bibfnamefont {T.}~\bibnamefont {van~der Lugt}}, \bibinfo {author} {\bibfnamefont {J.}~\bibnamefont {Barrett}},\ and\ \bibinfo {author} {\bibfnamefont {G.}~\bibnamefont {Chiribella}},\ }\bibfield  {title} {\bibinfo {title} {Device-independent certification of indefinite causal order in the quantum switch},\ }\href {https://doi.org/10.1038/s41467-023-40162-8} {\bibfield  {journal} {\bibinfo  {journal} {Nature Communications}\ }\textbf {\bibinfo {volume} {14}},\ \bibinfo {pages} {5811} (\bibinfo {year} {2023})}\BibitemShut {NoStop}%
\bibitem [{\citenamefont {Dong}\ \emph {et~al.}(2023)\citenamefont {Dong}, \citenamefont {Quintino}, \citenamefont {Soeda},\ and\ \citenamefont {Murao}}]{Dong2023quantumswitchis}%
  \BibitemOpen
  \bibfield  {author} {\bibinfo {author} {\bibfnamefont {Q.}~\bibnamefont {Dong}}, \bibinfo {author} {\bibfnamefont {M.~T.}\ \bibnamefont {Quintino}}, \bibinfo {author} {\bibfnamefont {A.}~\bibnamefont {Soeda}},\ and\ \bibinfo {author} {\bibfnamefont {M.}~\bibnamefont {Murao}},\ }\bibfield  {title} {\bibinfo {title} {The quantum switch is uniquely defined by its action on unitary operations},\ }\href {https://doi.org/10.22331/q-2023-11-07-1169} {\bibfield  {journal} {\bibinfo  {journal} {{Quantum}}\ }\textbf {\bibinfo {volume} {7}},\ \bibinfo {pages} {1169} (\bibinfo {year} {2023})}\BibitemShut {NoStop}%
\bibitem [{\citenamefont {Segal}(1947)}]{segal1947}%
  \BibitemOpen
  \bibfield  {author} {\bibinfo {author} {\bibfnamefont {I.~E.}\ \bibnamefont {Segal}},\ }\bibfield  {title} {\bibinfo {title} {Postulates for general quantum mechanics},\ }\href {http://www.jstor.org/stable/1969387} {\bibfield  {journal} {\bibinfo  {journal} {Annals of Mathematics}\ }\textbf {\bibinfo {volume} {48}},\ \bibinfo {pages} {930} (\bibinfo {year} {1947})}\BibitemShut {NoStop}%
\bibitem [{\citenamefont {Davies}\ and\ \citenamefont {Lewis}(1970)}]{Davies1970}%
  \BibitemOpen
  \bibfield  {author} {\bibinfo {author} {\bibfnamefont {E.~B.}\ \bibnamefont {Davies}}\ and\ \bibinfo {author} {\bibfnamefont {J.~T.}\ \bibnamefont {Lewis}},\ }\bibfield  {title} {\bibinfo {title} {An operational approach to quantum probability},\ }\href {https://doi.org/10.1007/BF01647093} {\bibfield  {journal} {\bibinfo  {journal} {Communications in Mathematical Physics}\ }\textbf {\bibinfo {volume} {17}},\ \bibinfo {pages} {239} (\bibinfo {year} {1970})}\BibitemShut {NoStop}%
\bibitem [{\citenamefont {Ludwig}(1968)}]{Ludwig1968}%
  \BibitemOpen
  \bibfield  {author} {\bibinfo {author} {\bibfnamefont {G.}~\bibnamefont {Ludwig}},\ }\bibfield  {title} {\bibinfo {title} {Attempt of an axiomatic foundation of quantum mechanics and more general theories. {III}},\ }\href {https://doi.org/10.1007/BF01654027} {\bibfield  {journal} {\bibinfo  {journal} {Communications in Mathematical Physics}\ }\textbf {\bibinfo {volume} {9}},\ \bibinfo {pages} {1} (\bibinfo {year} {1968})}\BibitemShut {NoStop}%
\bibitem [{\citenamefont {D{\"a}hn}(1968)}]{Dahn1968}%
  \BibitemOpen
  \bibfield  {author} {\bibinfo {author} {\bibfnamefont {G.}~\bibnamefont {D{\"a}hn}},\ }\bibfield  {title} {\bibinfo {title} {Attempt of an axiomatic foundation of quantum mechanics and more general theories. {IV}},\ }\href {https://doi.org/10.1007/BF01645686} {\bibfield  {journal} {\bibinfo  {journal} {Communications in Mathematical Physics}\ }\textbf {\bibinfo {volume} {9}},\ \bibinfo {pages} {192} (\bibinfo {year} {1968})}\BibitemShut {NoStop}%
\bibitem [{\citenamefont {Stolz}(1971)}]{Stolz1971}%
  \BibitemOpen
  \bibfield  {author} {\bibinfo {author} {\bibfnamefont {P.}~\bibnamefont {Stolz}},\ }\bibfield  {title} {\bibinfo {title} {Attempt of an axiomatic foundation of quantum mechanics and more general theories {VI}},\ }\href {https://doi.org/10.1007/BF01877753} {\bibfield  {journal} {\bibinfo  {journal} {Communications in Mathematical Physics}\ }\textbf {\bibinfo {volume} {23}},\ \bibinfo {pages} {117} (\bibinfo {year} {1971})}\BibitemShut {NoStop}%
\bibitem [{\citenamefont {Mielnik}(1974)}]{Mielnik1974}%
  \BibitemOpen
  \bibfield  {author} {\bibinfo {author} {\bibfnamefont {B.}~\bibnamefont {Mielnik}},\ }\bibfield  {title} {\bibinfo {title} {Generalized quantum mechanics},\ }\href {https://doi.org/10.1007/BF01646346} {\bibfield  {journal} {\bibinfo  {journal} {Communications in Mathematical Physics}\ }\textbf {\bibinfo {volume} {37}},\ \bibinfo {pages} {221} (\bibinfo {year} {1974})}\BibitemShut {NoStop}%
\bibitem [{\citenamefont {Giles}(1970)}]{Giles1970}%
  \BibitemOpen
  \bibfield  {author} {\bibinfo {author} {\bibfnamefont {R.}~\bibnamefont {Giles}},\ }\bibfield  {title} {\bibinfo {title} {Foundations for quantum mechanics},\ }\href {https://doi.org/10.1063/1.1665373} {\bibfield  {journal} {\bibinfo  {journal} {Journal of Mathematical Physics}\ }\textbf {\bibinfo {volume} {11}},\ \bibinfo {pages} {2139} (\bibinfo {year} {1970})}\BibitemShut {NoStop}%
\bibitem [{\citenamefont {Gudder}(1973)}]{Gudder1973}%
  \BibitemOpen
  \bibfield  {author} {\bibinfo {author} {\bibfnamefont {S.}~\bibnamefont {Gudder}},\ }\bibfield  {title} {\bibinfo {title} {Convex structures and operational quantum mechanics},\ }\href {https://doi.org/10.1007/BF01645250} {\bibfield  {journal} {\bibinfo  {journal} {Communications in Mathematical Physics}\ }\textbf {\bibinfo {volume} {29}},\ \bibinfo {pages} {249} (\bibinfo {year} {1973})}\BibitemShut {NoStop}%
\bibitem [{\citenamefont {Aubrun}\ \emph {et~al.}(2022)\citenamefont {Aubrun}, \citenamefont {Lami}, \citenamefont {Palazuelos},\ and\ \citenamefont {Pl\'avala}}]{PhysRevLett.128.160402}%
  \BibitemOpen
  \bibfield  {author} {\bibinfo {author} {\bibfnamefont {G.}~\bibnamefont {Aubrun}}, \bibinfo {author} {\bibfnamefont {L.}~\bibnamefont {Lami}}, \bibinfo {author} {\bibfnamefont {C.}~\bibnamefont {Palazuelos}},\ and\ \bibinfo {author} {\bibfnamefont {M.}~\bibnamefont {Pl\'avala}},\ }\bibfield  {title} {\bibinfo {title} {Entanglement and superposition are equivalent concepts in any physical theory},\ }\href {https://doi.org/10.1103/PhysRevLett.128.160402} {\bibfield  {journal} {\bibinfo  {journal} {Phys. Rev. Lett.}\ }\textbf {\bibinfo {volume} {128}},\ \bibinfo {pages} {160402} (\bibinfo {year} {2022})}\BibitemShut {NoStop}%
\bibitem [{\citenamefont {D'Ariano}\ \emph {et~al.}(2020)\citenamefont {D'Ariano}, \citenamefont {Erba},\ and\ \citenamefont {Perinotti}}]{PhysRevA.101.042118}%
  \BibitemOpen
  \bibfield  {author} {\bibinfo {author} {\bibfnamefont {G.~M.}\ \bibnamefont {D'Ariano}}, \bibinfo {author} {\bibfnamefont {M.}~\bibnamefont {Erba}},\ and\ \bibinfo {author} {\bibfnamefont {P.}~\bibnamefont {Perinotti}},\ }\bibfield  {title} {\bibinfo {title} {Classical theories with entanglement},\ }\href {https://doi.org/10.1103/PhysRevA.101.042118} {\bibfield  {journal} {\bibinfo  {journal} {Phys. Rev. A}\ }\textbf {\bibinfo {volume} {101}},\ \bibinfo {pages} {042118} (\bibinfo {year} {2020})}\BibitemShut {NoStop}%
\bibitem [{\citenamefont {Ac\'{\i}n}\ \emph {et~al.}(2010)\citenamefont {Ac\'{\i}n}, \citenamefont {Augusiak}, \citenamefont {Cavalcanti}, \citenamefont {Hadley}, \citenamefont {Korbicz}, \citenamefont {Lewenstein}, \citenamefont {Masanes},\ and\ \citenamefont {Piani}}]{PhysRevLett.104.140404}%
  \BibitemOpen
  \bibfield  {author} {\bibinfo {author} {\bibfnamefont {A.}~\bibnamefont {Ac\'{\i}n}}, \bibinfo {author} {\bibfnamefont {R.}~\bibnamefont {Augusiak}}, \bibinfo {author} {\bibfnamefont {D.}~\bibnamefont {Cavalcanti}}, \bibinfo {author} {\bibfnamefont {C.}~\bibnamefont {Hadley}}, \bibinfo {author} {\bibfnamefont {J.~K.}\ \bibnamefont {Korbicz}}, \bibinfo {author} {\bibfnamefont {M.}~\bibnamefont {Lewenstein}}, \bibinfo {author} {\bibfnamefont {L.}~\bibnamefont {Masanes}},\ and\ \bibinfo {author} {\bibfnamefont {M.}~\bibnamefont {Piani}},\ }\bibfield  {title} {\bibinfo {title} {Unified framework for correlations in terms of local quantum observables},\ }\href {https://doi.org/10.1103/PhysRevLett.104.140404} {\bibfield  {journal} {\bibinfo  {journal} {Phys. Rev. Lett.}\ }\textbf {\bibinfo {volume} {104}},\ \bibinfo {pages} {140404} (\bibinfo {year} {2010})}\BibitemShut {NoStop}%
\bibitem [{\citenamefont {Forster}\ \emph {et~al.}(2009)\citenamefont {Forster}, \citenamefont {Winkler},\ and\ \citenamefont {Wolf}}]{PhysRevLett.102.120401}%
  \BibitemOpen
  \bibfield  {author} {\bibinfo {author} {\bibfnamefont {M.}~\bibnamefont {Forster}}, \bibinfo {author} {\bibfnamefont {S.}~\bibnamefont {Winkler}},\ and\ \bibinfo {author} {\bibfnamefont {S.}~\bibnamefont {Wolf}},\ }\bibfield  {title} {\bibinfo {title} {Distilling nonlocality},\ }\href {https://doi.org/10.1103/PhysRevLett.102.120401} {\bibfield  {journal} {\bibinfo  {journal} {Phys. Rev. Lett.}\ }\textbf {\bibinfo {volume} {102}},\ \bibinfo {pages} {120401} (\bibinfo {year} {2009})}\BibitemShut {NoStop}%
\end{thebibliography}
\end{document}